\documentclass{article}
\usepackage{graphicx} 
\usepackage{natbib}
\usepackage{tikz}
\usepackage{tabularx}
\usepackage{tabularray}
\usepackage{doi}
\usepackage{url}
\usepackage{appendix}
\usepackage{comment}
\usepackage{amsfonts,amsmath,amssymb,amsthm}
\usepackage{authblk}

\newtheorem*{theorem*}{Theorem}
\newtheorem{theorem}{Theorem}[section]
\newtheorem{definition}{Definition}[section]

\newtheorem{corollary}{Corollary}[section]
\newtheorem{lemma}{Lemma}[section]

\setcitestyle{authoryear,open={(},close={)}}

\DeclareMathOperator{\Xt}{X_G^{(2)}}
\DeclareMathOperator{\dX}{d_{X}}
\DeclareMathOperator{\dXt}{d_{X}^{(2)}}

\DeclareMathOperator{\At}{\mathcal{A}_{G}^{(2)}}
\DeclareMathOperator{\dA}{d_{\mathcal{A}}}
\DeclareMathOperator{\dAt}{d_{\mathcal{A}}^{(2)}}

\usetikzlibrary{calc, shapes, fit,arrows.meta}

\title{Structured Coordination Games on Planar Graphs: a Dual Approach}
\author[1,2]{John S. McAlister}
\affil[1]{Princeton University, Program in Applied and Computational Mathematics}
\affil[2]{NSF Center for Analysis and Prediction of Pandemic Expansion (APPEX)}

\date{}

\begin{document}

\maketitle

\begin{abstract}
Finding Nash equilibria in the pure coordination game is trivial when every player interacts with every other player evenly. However, when players have a relational structure, the question becomes harder to answer. Here, we show that Nash equilibria to the structured coordination game represent a local notion of the Min-cut partition by way of the potential function. Using the standard relationship between partitions and subgraphs in the planar dual, we define a quasimetric space for which Nash equilibria are local minimizers of a global objective. When the game is restricted to only two strategies, these results allow us to construct a dynamic subgraph process which recapitulates the Myopic Best Response dynamics of the coordination game.    
\end{abstract}

\section{Introduction}
In the two-player coordination game, each player maximizes their payoff by playing the same strategy as their co-player. More specifically, these games are described by the bandwagon property of \cite{Kandori1998} which says that every player's best response is contained in the support of their coplayer's strategy. When we are restricted only to pure strategies, this property reduces to the simple case where every player is better off playing the same strategy as their opponent. In the two-player case, the game is fully understood, and in the multiplayer case, if every player is interacting with every other player evenly, this game is equally easy to analyze. The only $\epsilon-$Nash equilibria are the so-called consensus equilibria where every player is playing the same strategy \citep{Kandori1993,Kandori1995}. 

If there is an inhomogeneous relational structure, however, the game becomes more complicated to analyze. In this case, we imagine the game being played on a graph where individuals are vertices, and they share edges with other players with whom they interact. There has been a lot of work done on this type of system from an economic game theoretic perspective \citep{Ellison1993,Ellison2000,Gilboa1991,Ely2002,Oechssler1997,Oechssler1999,Robson1995}. The main results from these studies, which are well reviewed in \cite{Weidenholzer}, show that with specific relational structures, like a lattice, cycle, or path, the consensus equilibria are still the only $\epsilon-$Nash equilibria. However, the relational structure has an effect on equilibrium selection. This means that, under myopic best response or similar strategy revision protocols, certain relational structures can promote the Pareto-dominant Nash equilibrium or the risk-dominant Nash equilibrium. 

This structured coordination game has also been studied extensively from a statistical physics and sociology perspective under the name ``opinion dynamics" \citep{Degroot1974,friedkin1990social}. These models share a great deal of similarity, and their analysis has a similar history.  

More recently, there has been work on general graphs which do not rely on symmetry. Conditions for existence, stability \citep{Arditti2024}, and invasibility \cite{Paarporn2021} of consensus equilibria have been examined analytically in the case of only two strategies. Non-consensus equilibria (those Nash equilibria where more than one strategy is expressed) are harder to analyze because of their extreme dependence on the structure of the graph. Much of the recent work in understanding these equilibria has come from simulation studies like \cite{Buskens2016,Tomassini2010,McAlister2024simulation} which show how features of the graph like connectedness and diameter change the likelihood of converging to a consensus or non-consensus equilibrium under some strategy revision protocol, but simulation cannot provide analytical results directly relating these non-consensus equilibria to the relational structures which admit them. 

These non-consensus equilibria are of interest because they are related to questions of community detection. There has been a great focus on community detection in the past several decades, with particular interest in top-down vertex partitioning methods like modularity \cite{Clauset2004,Newman2006}. In the same way that a maximum modularity partition of a graph partitions the graph into well-connected communities which are themselves sparsely connected, a non-consensus equilibrium in a coordination game partitions the graph into strategic communities which are more connected internally than externally. 

In this manuscript, we formalize the relationship between the coordination game with payoff matrix $I_m$ and partitioning approaches involved in community detection and proceed to prove analytical results about characteristics of the non-consensus equilibria in planar graphs. We start in section \ref{sec:paritioning} by defining a vertex partition which is equivalent to a non-consensus equilibrium and comparing it to other types of partitions. Then, in section \ref{sec:nashequilibria}, we show that finding such a partition is equivalent to finding a local minimizer of size among subgraphs of the planar dual.  This method can be used to understand qualities of graphs which admit non-consensus equilibria, but it cannot be used to understand the dynamics of coordinating systems under myopic best response unless we follow in the footsteps of previous work and restrict to only two strategies, which we do in section \ref{sec:mbr}. we conclude by discussing how this dual approach improves our understanding of coordination, coordinating systems, and community detection in section \ref{sec:discussion}.

\section{Partitioning and Equilibria}\label{sec:paritioning}
\subsection{The Equilibrium Partition}
Consider the game with payoff matrix $A=I_m$. In the two-player case, if two players use the same strategy, both get a payoff of $1$, and if they use different strategies, they get a payoff of $0$. Assume that a player picks a strategy from the set $C=\{1,...,m\}$ This way a player's payoff is given as 
\[w(i|j)=\delta(i,j)\]
where $\delta(i,j)$ is the Kronecker-$\delta$ which is 1 when $i=j$ and $0$ otherwise. In a multiplayer game, we may say that every player's total payoff is the sum of the pairwise payoffs it gets from each diadic interaction. When the game is played on the graph $G(V,E)$ with adjacency matrix $W$, we let $u:V\to C$ be a strategy profile so the payoff for player $v$ playing against the strategy profile $u$ is computed as 
\[w_v(i|u)=\sum_{w\in V}W_{v,w}\delta(i,u_w)\]
It is important to notice that, because the payoff matrix is the identity, any permutation of the strategies would leave the payoff of every player unchanged. (This result is proved in \cite{McAlister2024simulation} and in \cite{McAlister2026}, but it is a straightforward enough idea that we will not prove it here.) Because of this fact, we will group strategy profiles into equivalence classes of vertex partitions.

Define $\Phi$ as a function from the set of strategy profiles $C^n$ with $n$ players to the class of vertex partitions $\mathcal{Q}_G$ on the graph $G$. It is entirely defined by saying that if $v,w\in V$ use the same strategy in $u$ then $v$ and $w$ are in the same part of the partition $\Phi u$. This mapping is, of course, not injective because every strategic permutation of $u$ maps to the same vertex partition under $\Phi$, so we will discuss only equivalence classes under $\Phi$. we call the set of equivalence classes under $\Phi$, $\mathcal{A}_G$. It is now obvious that $\Phi:\mathcal{A}_G\to \mathcal{Q}_G$ is a bijection. 

Every vertex partition can be considered as a strategy profile, but we are most interested in those vertex partitions which correspond to Nash equilibria.
\begin{definition}
    $P\in \mathcal{Q}_G$ is an \textit{Equilibrium Partition} if and only if there is a $u\in C^n$ such that $\Phi u = P$ and $u$ is a Nash equilibrium. 
\end{definition}

Because all strategy profiles in the same equivalence class under $\Phi$ result in the same payoffs for every player, we call $a\in \mathcal{A}$ a Nash equilibrium if and only if at least one of the elements (and thus all of the elements) in $a$ is a Nash equilibrium. This is a slight abuse of language, but such an abuse makes the system easier to discuss. With this understanding, we can equivalently define an equilibrium partition as any $P\in \mathcal{Q}_P$ such that $\Phi^{-1}P$ is a Nash equilibrium. As a matter of notation. We will now simply call an element of $\mathcal{A}$ a strategy profile and use $u$ to denote such an object. 

The main characteristic of an equilibrium partition is that every vertex is in the part of the partition that contains the plurality of its neighbors. That is equivalent to saying that, for player $v$ in part $i$ of the partition, no other part in the partition contains more of the neighbors of player $v$ than part $i$ contains. Examples of trivial equilibrium partitions (corresponding to the consensus equilibrium) and non-trivial equilibrium partitions are included in figure \ref{fig:examples} 

\begin{figure}
    \centering
        \begin{tikzpicture}[scale = 0.9]
            \node(a)[circle, fill=orange, inner sep =2pt] at (6.5+0,0.5+0){};
            \node(b)[circle, fill=orange, inner sep =2pt] at (6.5+1,0.5+0){};
            \node(c)[circle, fill=orange, inner sep =2pt] at (6.5+1,0.5+1){};
            \node(d)[circle, fill=orange, inner sep =2pt] at (6.5+1,0.5+-1){};
            \node(e)[circle, fill=orange, inner sep =2pt] at (6.5+2,0.5+0){};
            \node(f)[circle, fill=blue, inner sep =2pt] at (6.5+3,0.5+0){};
            \node(g)[circle, fill=blue, inner sep =2pt] at (6.5+3,0.5+1){};
            \node(h)[circle, fill=blue, inner sep =2pt] at (6.5+3,0.5+-1){};
            \node(i)[circle, fill=blue, inner sep =2pt] at (6.5+4,0.5+0){};

            \node[fit=(a)(b)(c)(d)(e),dashed, draw, thick, rectangle,rounded corners=10,inner sep=4pt] {};
            \node[fit=(f)(g)(h)(i),dashed, draw, thick, rectangle,rounded corners=10,inner sep=4pt] {};

            \draw[line width = 0.5mm](a)--(b);
            \draw[line width = 0.5mm](a)--(c);
            \draw[line width = 0.5mm](a)--(d);
            \draw[line width = 0.5mm](b)--(c);
            \draw[line width = 0.5mm](b)--(d);
            \draw[line width = 0.5mm](b)--(e);
            \draw[line width = 0.5mm](c)--(e);
            \draw[line width = 0.5mm](d)--(e);
            \draw[line width = 0.5mm](e)--(f);
            \draw[line width = 0.5mm](e)--(g);
            \draw[line width = 0.5mm](e)--(h);
            \draw[line width = 0.5mm](f)--(g);
            \draw[line width = 0.5mm](f)--(h);
            \draw[line width = 0.5mm](f)--(i);
            \draw[line width = 0.5mm](g)--(i);
            \draw[line width = 0.5mm](h)--(i);
            \draw[line width = 0.5mm](a)to [out = 90, in = 90, looseness =1.5](i);

        \end{tikzpicture}
        \hspace{1cm}
        \begin{tikzpicture}[scale = 0.9]
         \node(a)[circle, fill=blue, inner sep =2pt] at (6.5+0,0.5+0){};
            \node(b)[circle, fill=blue, inner sep =2pt] at (6.5+1,0.5+0){};
            \node(c)[circle, fill=blue, inner sep =2pt] at (6.5+1,0.5+1){};
            \node(d)[circle, fill=blue, inner sep =2pt] at (6.5+1,0.5+-1){};
            \node(e)[circle, fill=blue, inner sep =2pt] at (6.5+2,0.5+0){};
            \node(f)[circle, fill=blue, inner sep =2pt] at (6.5+3,0.5+0){};
            \node(g)[circle, fill=blue, inner sep =2pt] at (6.5+3,0.5+1){};
            \node(h)[circle, fill=blue, inner sep =2pt] at (6.5+3,0.5+-1){};
            \node(i)[circle, fill=blue, inner sep =2pt] at (6.5+4,0.5+0){};
            \node(inv) at (6.5+4,2.3){};

            \node[fit=(a)(b)(c)(d)(e)(f)(g)(h)(i)(inv),dashed, draw, thick, rectangle,rounded corners=10,inner sep=5pt] {};

            \draw[line width = 0.5mm](a)--(b);
            \draw[line width = 0.5mm](a)--(c);
            \draw[line width = 0.5mm](a)--(d);
            \draw[line width = 0.5mm](b)--(c);
            \draw[line width = 0.5mm](b)--(d);
            \draw[line width = 0.5mm](b)--(e);
            \draw[line width = 0.5mm](c)--(e);
            \draw[line width = 0.5mm](d)--(e);
            \draw[line width = 0.5mm](e)--(f);
            \draw[line width = 0.5mm](e)--(g);
            \draw[line width = 0.5mm](e)--(h);
            \draw[line width = 0.5mm](f)--(g);
            \draw[line width = 0.5mm](f)--(h);
            \draw[line width = 0.5mm](f)--(i);
            \draw[line width = 0.5mm](g)--(i);
            \draw[line width = 0.5mm](h)--(i);
            \draw[line width = 0.5mm](a)to [out = 90, in = 90, looseness =1.5](i);

        \end{tikzpicture}
    \caption{\textbf{Left} a non-consensus Nash equilibrium and associated non-trivial partition. \textbf{Right} a consensus Nash equilibrium and the associated trivial partition.}
    \label{fig:examples}
\end{figure}
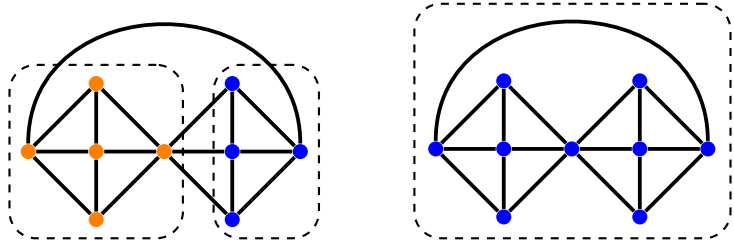

Because of this property of equilibrium partition, it can be thought of as a useful way to discuss community detection in a more ``locally driven" sense. Each individual is deciding which part of the partition it ``wants" to belong to, and an equilibrium partition is a partitioning of the vertices so that no one can improve their payoff. This means that in an equilibrium partition, every player is at least as connected within their part as they are to any other part in the partition. 

With this definition of equilibrium partitions, we must now note that the coordination game at hand is a potential game. A potential game, first described by \cite{Monderer1996} is game for which there is a potential function $\mathcal{W}:C^n\to \mathbb{R}$ so that for all $v$
    \begin{equation}\label{eq:PotentialDef}\mathcal{W}(r,u_{-v})-\mathcal{W}(s,u_{-v})=w_{v}(r|u_{-v})-w_v(s|u_{-v})\end{equation}
    
    where $u_{-v}$ is the strategy of every player except player $v$. Thus, $(r,u_{-v})$ is the strategy profile where player $v$ plays strategy $r$ and every other player plays according to the strategy profile $u_{-v}$.
    In plain English, this means that there is a single function which describes the change in payoff from a unilateral change in strategy from any player. More specifically, equation \eqref{eq:PotentialDef} describes an exact potential function. Other kinds of potential function exist \citep{Quang2016} and are useful, but for the coordination game, the exact potential function will suffice. The most important point about the potential function is that if players sequentially change their strategies to improve their own payoff, they form a sequence of strategy profiles which will certainly lead to a Nash equilibrium. These paths (called finite improvement paths) imply that, for continuous strategy spaces, Nash equilibria are local maxima of the potential function. 

\begin{lemma}\label{lem:potentialFunction}
    For the coordination game with payoff matrix $I_m$ played of the graph $G$ with adjacency matrix $W$, \[\mathcal{W}(u)=\frac{1}{2}\sum_{v\in V}w(u_v,u)\] is an exact potential function
\end{lemma}
\begin{proof}
    The proof is a direct computation. Consider the strategy profile $u$ in which a focal individual $x$ plays the strategy $s$ and the strategy profile $u'$ in which every player plays the same strategy as in $u$ except for the player $x$ who plays strategy $r$. This means that 
    \begin{equation*}
        \begin{split}
            2\mathcal{W}(u)-2\mathcal{W}(u')&=\sum_{v\in V}w(u_v,u)-w(u'_v,u')\\
            &=\sum_{v\in V}\sum_{w\in V}W_{v,w}(\delta(u_v,u_w)-\delta(u'_v,u'_w))
        \end{split}
    \end{equation*}
    If neither $v$ nor $w$ is player $x$ then $\delta(u_v,u_w)-\delta(u'_v,u'_w)=0$. Also, if both $v$ and $w$ are the player $x$ then $W_{v,w}=W_{x,x}=0$ because the graph is simple with no loops. Therefore we have the equation
    \begin{equation*}
        \begin{split}
            2\mathcal{W}(u)-2\mathcal{W}(u')&=\sum_{v\neq x}W_{v,x}\delta(u_v,s)+\sum_{w\neq x} W_{x,w}\delta (s,u_w)\\
            &\quad\quad\quad\quad-\sum_{v\neq x}W_{v,x}\delta(u'_v,r)-\sum_{w\neq x} W_{x,w}\delta (r,u'_w)
        \end{split}
    \end{equation*}
    Because of the symmetry of $W$ and of $\delta$ we can combine the sums to see that 
    \[2\mathcal{W}(u)2-\mathcal{W}(u')=2\sum_{v\neq x}W_{v,x}\delta(u_v,s)-2\sum_{v\neq x}W_{v,x}\delta(u'_v,r)\]

    Finally note that for $v\neq x$, $u_v=u'_v$ so we have the final expression
    \begin{equation*}
        \begin{split}\mathcal{W}(u)-\mathcal{W}(u')&=\sum_{v\neq x}W_{x,w}\delta(s,u_w)-\sum_{w\neq x} W_{x,w}\delta (r,u_w)\\
        &= w(s|u_{-x})-w(r|u_{-x})
        \end{split}
    \end{equation*}
    This is exactly the condition required to be an exact potential function. 
\end{proof}

Given that the game is a potential game, we know that Nash equilibria are local maximizers of the potential function. Therefore, equilibrium partitions correspond to local maximizers of the potential function. Of course, we must clarify what is meant by ``local." The description of ``locality" and its consequences are discussed later in this section and in section \ref{sec:nashequilibria}. First we describe how the potential function can be used to unite our understanding of the equilibrium partition with other partitions commonly used in community detection.  

\subsection{Relation to other vertex partitions}

The most famous community detection partitioning method is the maximum modularity partition, which, unsurprisingly, seeks to find the partition $P$ which maximizes the quantity 
\[Q(P)=\frac{1}{2m}\sum_{v,w\in V^2}\left[W_{v,w}-\frac{d_vd_w}{2m}\right]\delta(c_v,c_w)\]
Where $W$ is the adjacency matrix, $d_v$ is the degree of $v$ and $\delta(c_v,c_w)$ is a function which returns 1 if $v$ and $w$ are in the same part in $P$ and $0$ otherwise \citep{Clauset2004}.

Typically, when network scientists are discussing community detection, they are describing methods which find or approximate maximum modularity partitions. Finding such a partition is known to be $\mathcal{NP}$-complete \citep{Brandes2006} and so for very large graphs, approximation methods must be used.

In a heuristic way, the modularity partitions and the equilibrium partitions have a great deal of overlap, and indeed, many partitions which are maximum modularity partitions are also equilibrium partitions. However, there are examples of maximum modularity partitions which are not equilibrium partitions, and there are examples of equilibrium partitions which are not maximum modularity partitions (Fig. \ref{fig:ModularityvsEquilibrium}).  
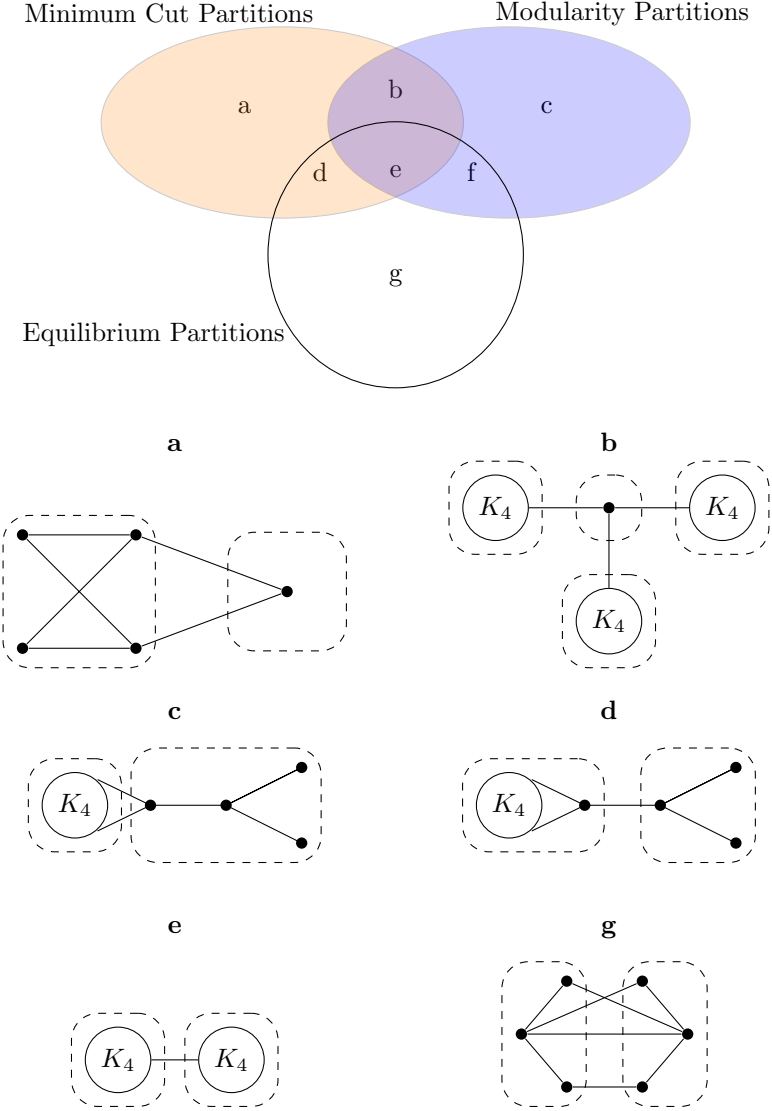
\begin{figure}
    \centering
    \begin{tikzpicture}
		\node (E1) at (0,0){};
		\node (E2) at (0,1){};
		\node (EQU) at (3,-0.4){};
		\node (Q1) at (6,0){};
		\node (Q2) at (6,1){};
		\node (U1) at (4,-2.5){};
		\node (U2) at (2,-2.5){};
		\node (U3) at (3,-3){};
		
		\node(a) at (1,0.5){a};
		\node(b) at (3,0.75){b};
		\node(c) at (5,0.5){c};
		\node(e) at (3,-0.35){e};
		\node(d) at (2,-0.35){d};
		\node(f) at (4, -0.35){f};
		\node(g) at (3, -1.75){g};

		\node[fit=(E1)(EQU)(E2), draw, ellipse, fill = orange, opacity = 0.2 ,inner sep=2pt]{};
		\node[fit=(Q1)(EQU)(Q2), draw, ellipse,fill = blue, opacity = 0.2, inner sep=2pt]{};
		\node[fit=(U1)(EQU)(U2), draw, ellipse,inner sep=2pt]{};
		
		\node (Ename) at (0,1.75){Minimum Cut Partitions};
		\node(Qname) at (6,1.75){Modularity Partitions};
		\node(Uname) at (-0.2,-2.5){Equilibrium Partitions};
	\end{tikzpicture}\\
	\vspace{0.5cm}
	\begin{tabular}{ccc}
		
		\textbf{a}&&\textbf{b}\\
		
		\begin{tikzpicture}
			\node(a) [circle, fill, inner sep =1.5pt] at (0,0){};
			\node(b) [circle, fill, inner sep =1.5pt] at (1.5,0){};
			\node(c) [circle, fill, inner sep =1.5pt] at (0,1.5){};
			\node(d) [circle, fill, inner sep =1.5pt] at (1.5,1.5){};
			\node(e) [circle, fill, inner sep =1.5pt] at (3.5,0.75){};
			
			\draw(b)--(a)--(d)--(c)--(b)--(e)--(d);
			
			\node[fit=(a)(b)(c)(d),dashed,draw, rectangle,rounded corners=10,inner sep=5pt] {};
			\node[fit=(e),dashed, draw, rectangle,rounded corners=10,inner sep=20pt] {};

		\end{tikzpicture}&&\begin{tikzpicture}
			\node (A) [draw = black, circle] at (0,0){$K_4$};
			\node (B) [draw = black, circle] at (3,0){$K_4$};
			\node (C) [draw = black, circle] at (1.5,-1.5){$K_4$};
			\node(a)[circle, fill, inner sep =1.5pt] at (1.5,0){};
			
			\draw (A)--(a)--(B);
			\draw (a)--(C);
			
			\node[fit=(A),dashed,draw, rectangle,rounded corners=10,inner sep=5pt] {};
			\node[fit=(B),dashed, draw, rectangle,rounded corners=10,inner sep=5pt] {};
			\node[fit=(C),dashed, draw, rectangle,rounded corners=10,inner sep=5pt] {};
			\node[fit=(a),dashed, draw, rectangle,rounded corners=10,inner sep=10pt] {};
			
		\end{tikzpicture}\vspace{0.25cm}\\
		\textbf{c}&&\textbf{d}\vspace{0.25cm}\\
		\begin{tikzpicture}
			\node (A) [draw = black, circle] at (1,0){$K_4$};
			\node (B) [circle,fill,inner sep=1.5pt]at (2,0){};
			\node (c) [circle,fill,inner sep=1.5pt]at (3,0){};
			\node (d) [circle,fill,inner sep=1.5pt]at (4,0.5){};
			\node (e) [circle,fill,inner sep=1.5pt]at (4,-0.5){};
			\node[fit=(A),dashed,draw, rectangle,rounded corners=10,inner sep=5pt] {};
			\node[fit=(B)(c)(d)(e),dashed, draw, rectangle,rounded corners=10,inner sep=5pt] {};
			
			\draw (1.3,0.34)--(2,0)--(1.3,-0.34);
			\draw(2,0)--(3,0)--(4,0.5)--(3,0)--(4,-0.5);
		\end{tikzpicture}
		&\hspace{0.5cm} &
		\begin{tikzpicture}
			\node (A) [draw = black, circle] at (1,0){$K_4$};
			\node (B) [circle,fill,inner sep=1.5pt]at (2,0){};
			\node (c) [circle,fill,inner sep=1.5pt]at (3,0){};
			\node (d) [circle,fill,inner sep=1.5pt]at (4,0.5){};
			\node (e) [circle,fill,inner sep=1.5pt]at (4,-0.5){};
			\node[fit=(A)(B),dashed,draw, rectangle,rounded corners=10,inner sep=5pt] {};
			\node[fit=(c)(d)(e),dashed, draw, rectangle,rounded corners=10,inner sep=5pt] {};
			
			\draw (1.3,0.34)--(2,0)--(1.3,-0.34);
			\draw(2,0)--(3,0)--(4,0.5)--(3,0)--(4,-0.5);
		\end{tikzpicture}\vspace{0.5cm}\\
		
		\textbf{e}&&\textbf{g}\vspace{0.25cm}\\
		\begin{tikzpicture}
			\node (A) [draw = black, circle] at (0,0){$K_4$};
			\node (B) [draw = black, circle] at (1.5,0){$K_4$};
			
			\draw(A)--(B);
			
			\node[fit=(A),dashed,draw, rectangle,rounded corners=10,inner sep=5pt] {};
			\node[fit=(B),dashed,draw, rectangle,rounded corners=10,inner sep=5pt] {};
		\end{tikzpicture}&& \begin{tikzpicture}
    		\node(a)[circle, fill, inner sep =1.5pt] at (0,-0.2){};
            \node(b)[circle, fill, inner sep = 1.5pt] at(0,1.2){};
            \node(c)[circle, fill, inner sep = 1.5pt] at(1,-0.2){};
            \node(d)[circle, fill, inner sep = 1.5pt] at(1,1.2){};
            \node(e)[circle, fill, inner sep = 1.5pt] at(1.6,0.5){};
            \node (f)[circle, fill, inner sep = 1.5pt] at(-0.6,0.5){};

            \draw(f)--(a)--(c)--(e)--(f)--(d)--(e);
            \draw(e)--(b)--(f);
        
            \node[fit=(e)(c)(d),dashed, draw, rectangle,rounded corners=10,inner sep=5pt] {};
            \node[fit=(a)(b)(f),dashed, draw, rectangle,rounded corners=10,inner sep=5pt] {};
            
\end{tikzpicture} \vspace{0.5cm}
		
	\end{tabular}
	\caption[Comparison of mincut, modularity, and equilibrium partitions]{\textbf{Top} A Venn diagram showing Minimum Cut Partitions, Modularity Partitions, and Equilibrium Partitions. \textbf{Bottom} In each of the 6 panels below the Venn diagram is an example (if one is known) of a partition which falls into each region of the Venn diagram. Clusters in the partitions are marked out by dashed lines. For the region $f$, no example is given here, although there is no reason to suspect this region is empty.}
    \label{fig:ModularityvsEquilibrium}
\end{figure}

Although they are not identical, equilibrium partitions and modularity partitions share a great similarity. In each, every vertex is more connected in its community than to any other community. This relationship is easy to see when we consider $k-$regular graphs. 

with added regularity, the modularity of the partition $P$ can be written more neatly as
\begin{equation}\label{eq:regularMod1}
    \begin{split}
    Q(P)=&\frac{1}{2m}{\sum_{v\in V}\sum_{w\in V}\left[W_{v,w}-\frac{k^2}{2m}\right]}\delta(c_v,c_w)\\
    &\frac{1}{2m}{\sum_{v\in V}\sum_{w\in V}\left[W_{v,w}\delta(c_v,c_w)-\frac{k^2}{2m}\delta(c_v,c_w)\right]}
    \end{split}
\end{equation}
Let $u\in \Phi^{-1}P\in \mathcal{A}$ so that we can express $Q(P)$ in terms of $\mathcal{W}(\Phi^{-1}P)$.
splitting the sum in \eqref{eq:regularMod1} we find that 
\begin{equation}\label{eq:regularMod2}
    \begin{split}
        Q(P)&=\frac{1}{2m}{\sum_{v\in V}\sum_{w\in V}\left[W_{v,w}\delta(c_v,c_w)\right]-\frac{k^2}{2m}\sum_{v\in V}\sum_{w\in V}\left[\delta(c_v,c_w)\right]}\\
        &=\frac{1}{m}\frac{1}{2}\sum_{v\in V}w_v(u_v|u)-\frac{k^2}{(2m)^2}\sum_{v\in V}|c_v|
    \end{split}
\end{equation}
where $|C_v|$ is the number of vertices in the same part of the partition as player $v$. Notice the first term in \eqref{eq:regularMod2} is exactly the potential function from lemma \ref{lem:potentialFunction}. Also notice that $\sum_{v\in V}|c_v|=\sum_{i\in C}|c_i|^2$ where $C$ enumerates the parts of the partition $P$ and $|c_i|$ is the number of vertices in the $i$th part of the partition. This means we have
\[
    Q(P)=\frac{1}{m}\mathcal{W}(\Phi^{-1}P)-\frac{k^2}{(2m)^2}\sum_{i\in C}|c_i|^2\]

Lastly, because of the regularity of the graph we can write $2m=nk$ and so our final expression is 
\begin{equation}
    Q(P)=\frac{1}{m}\mathcal{W}(\Phi^{-1}P)-\frac{1}{n^2}\sum_{i\in C}{|c_i|^2}
\end{equation}

Having made this connection, at least in the setting with regular graphs, we can see that maximizing modularity is akin to maximizing game potential with the additional consideration of minimizing the sum of squares of the part sizes in the partition. 

This demonstrates the connection between the two kinds of partitions but it also shows a key limitation in translating between modularity and equilibrium partitions. To calculate the game potential, each player need only know the state of the individuals in its neighborhood and so local changes can be made to improve potential and move towards a local maximum. For the modularity partition, additional global knowledge of the partition is needed.  Because the equilibrium partition is somehow ``less constrained" than the modularity partition,  one can generate pathological equilibrium partitions on amodular graphs. (e.g., $C_4$ as in fig \ref{fig:amodular}) however, the equilibrium partition is associated with a non-strict Nash equilibrium, leading to the conjecture that an amodular graph cannot have an equilibrium partition associated with a strict Nash equilibrium \citep{McAlister2026}.
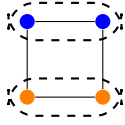
\begin{figure}
    \centering
    \begin{tikzpicture}
            \node(a)[circle, fill=orange, inner sep =2pt] at (0,0){};
            \node(b)[circle, fill=orange, inner sep =2pt] at (1,0){};
            \node(c)[circle, fill=blue, inner sep =2pt] at (1,1){};
            \node(d)[circle, fill=blue, inner sep =2pt] at (0,1){};

            \draw(a)--(b)--(c)--(d)--(a);

            \node[fit=(a)(b),dashed, draw, thick, rectangle,rounded corners=10,inner sep=4pt] {};

            \node[fit=(c)(d),dashed, draw, thick, rectangle,rounded corners=10,inner sep=4pt] {};
            
    \end{tikzpicture}
    \caption{$C_4$ is an amodular graph meaning that the maximum modularity partition is the trivial partition. The partition pictured here is an equilibrium partition representing a Nash equilibrium but it is not a strict Nash equilibrium. }
    \label{fig:amodular}
\end{figure}

The other important kind of vertex partition with which our partitioning method shares great similarity is the min-cut partition. This vertex partition is the partition which requires the fewest number of edges that connect vertices in different parts of the partition. 

Let $W$ be the adjacency matrix for the graph $G$, which has $m$ edges. The minimum edge cut partition is the partition into $|C|$ parts which minimizes the quantity $E(u)$. 
		\begin{equation}\label{eq:mincut}
			E(u) = \sum_{v,w}W_{vw}(1-\delta(u_vu_w))= 2m-\sum_{v,w}W_{vw}\delta(u_v,u_w)	\end{equation}

    From the form of the mincut partition objective functional \eqref{eq:mincut} it is easy to see how the mincut partition is related to the equilibrium partition. 
    \begin{corollary}\label{cor:mincutPotential}
        $E(\Phi u)= 2m-\mathcal{W}(u)$ and so $-\frac{1}{2}E(\Phi u)$ is an exact potential function for the coordination game.
    \end{corollary}

    \begin{proof}
        The result is obvious once we note that a potential function plus a constant is still a potential function. Now we need only observe that $\frac{-1}{2}E(\Phi u)=\mathcal{W}(u)-m$ from Lemma \ref{lem:potentialFunction}.
    \end{proof}

    This means that minimizing $E$ is the same as maximizing $\mathcal{W}$. The only difference between the two types of partitions is the constraints and the type of optimization. A mincut partition is the global minimizer of $E$ (and thus the global maximizer of $\mathcal{W}$) subject to the constraint there there are $|C|$ parts in the partition. In contrast an equilibrium partition is any local maximizer of $\mathcal{W}$ without the constraint on the number of parts in the partition. For this reason we say that an equilibrium partition is the local version of a mincut partition.  

    This tells us something about the overlap between mincut and equilibrium partitions. In particular we know that any mincut partition which is not on the ``boundary" of its constrained space of partitions (meaning that it does not have a part of size 1) must also be an equilibrium partition. 

    We have thus shown that finding equilibrium partitions not only tells us something important about the coordination game, it also tells us something about community detection. Equilibrium partitions can be thought of as a type of locally driven type of community detecting partition and they are associated to other community detecting partitions through the potential function. 

    \subsection{Defining locality in the space of partitions}
    In order to find local maximizers of game potential, we must first describe what ``local" means. The notion of locality we need is defined by the unilateral decisions each player (vertex) can make. From a strategy profile $u$, each of the $n$ players in that strategy profile can change their strategy to one of $|C|-1$ different strategies which means that there are $n(|C|-1)$ strategy profiles that are only 1 unilateral decision away from $u$. Any two strategy profiles can be connected by a chain of unilateral decisions by a finite number of players, and so we can use the unit of ``single strategic changes" to put a metric on the space of strategy profiles. More specifically, let $X$ be the space of strategy profiles and let $d_X(x,y)$ be a metric on $X$ equal to the minimum number of single strategic changes necessary to get from strategy $x$ to strategy $y$. If we think of every strategy profile in $X$ as a vertex on a graph where two vertices are connected by an edge if they are a single strategic change away from one another, then the metric is exactly the graph distance between vertices $x$ and $y$. 
    \begin{definition}[Strategy space $(X_G,d_X(\cdot,\cdot))$]
        The strategy space $(X,d_X(\cdot,\cdot))$ is a metric space where $X$ is the set of all strategy profiles $(C^n)$ on the graph $G$, and $d_X(x,y)$ is the minimum number of single strategic changes to get from strategy profile $x$ to strategy profile $y$.
    \end{definition}
    A Nash equilibrium is any strategy profile $u\in X$ so that $\mathcal{W}(u)\geq\mathcal{W}(x)$ for all $x\in \Gamma(u):=\{y\in X,d(y,u)\leq 1\}$.

    Because $\mathcal{W}=-\frac{1}{2}E(\Phi(u))$ is an exact potential function, when two strategy profiles are in the same equivalence class under $\Phi$, they must have the same value of $\mathcal{W}$. Therefore, we define a similar notion of locality for the set of equivalence classes $\mathcal{A}$. Two equivalence classes $A,B\in \mathcal{A}$ are adjacent to one another if any two strategy profiles $a\in A, b\in B$ are adjacent to one another in $(X,d(\cdot,\cdot))$. Therefore, in the same way, we have a metric on the set $\mathcal{A}$. 
    \begin{definition}[Reduced strategy space $(\mathcal{A},d_\mathcal{A}(\cdot,\cdot))$]
        The strategy space $(\mathcal{A}_G,d_\mathcal{A}(\cdot,\cdot))$ is a metric space where $\mathcal{A}$ is the set of all equivalence classes of strategy profiles under $\Phi$ on the graph $G$, and metric $\dA$ which is defined by the adjacency given by \[\dA(a,b)=1\iff \exists (u,u')\in (a,b) \text{ such that }\dX(u,u')=1\]
    \end{definition}
    Moreover, we know that $u\in \mathcal{A}$ is an equilibrium partition if $\mathcal{W}(u)\geq\mathcal{W}(x)$ for all $x\in \Gamma_\mathcal{A}(u):=\{y\in \mathcal{A};d_{\mathcal{A}}(y,u)\leq 1\}.$ 

    Although this reframing allows us to unite the ideas of mincut partitions and equilibrium partitions very naturally, the metric space $(\mathcal{A}_G,d_\mathcal{A})$ is of little use without understanding its structure. Moreover, the space $(\mathcal{A}_G,d_\mathcal{A})$ does not much improve the way that we can understand the system geometrically. As with most questions about the coordination game, we cannot, at present, describe the structure of the metric space in general, nor can we give a general geometric interpretation. However, we can give a geometric interpretation for a certain subset of graphs, the planar graphs. In this setting, we can describe a quasimetric space for which our understanding of locality can be described using only the partition boundaries. In section \ref{sec:nashequilibria}, we will describe a new quasimetric space with some of the same properties as $(\mathcal{A}_G,d_{\mathcal{A}})$ on this set of graphs which allows us to describe equilibrium partitions through a single minimization problem rather than many coupled maximization problems. From this, we arrive at an equilibrium result, but we cannot derive any results about the game dynamics in general for reasons discussed later in the chapter. Only in the case that there are only two strategies (section \ref{sec:mbr}) can we use the dual understanding of the problem in the dynamic sense. Through these dual approaches, we will gain an understanding of the geometry of coordination and local community detection.

\section{Nash equilibria on planar graphs}\label{sec:nashequilibria}
The results about the potential game are easy to state and to understand in the reduced strategy space, but they are hard to use because the structure of the space is hard to imagine, and the potential surface on that space can feel nonsensical. The goal now is to show, through quasiisometry, that the reduced strategy space is related to an easier to imagine (although harder to define) space so that the potential game results can be used to gain some geometric understanding of the system. To begin, we will consider a quasimetric space called the bridge-free subgraph space. The main result of this section is to show that subgraphs which locally minimize size in this quasimetric space are associated with Nash equilibria in the Strategy space. 

\subsection{bridge-free subgraph space}
The bridge-free subgraph space (which we will call simply the subgraph space) is a quasimetric space which we will build on the set of all bridge-free subgraphs of a planar graph $G$. Call this set $\mathcal{SC}_{G}$. To equip this set with a quasimetric, we must introduce several definitions. 

\begin{definition}[Departure vertex]
    For a face $f$ of the planar graph $G$ and a subgraph $s\in \mathcal{SC}_G$, a vertex $v$ adjacent to the face $f$ is a departure vertex of $f$ if $s$ includes $v$ and at least one edge adjacent to $v$ and not incident to the face $f$. (e.g., Fig. \ref{fig:departureVertex})
\end{definition}

\begin{figure}
    \centering
    \begin{tikzpicture}
            \node(a)[circle, fill = red, inner sep =1.5pt] at (0,0){};
            \node(b)[circle, fill, inner sep =1.5pt] at (0.5,0.866){};
            \node(c)[circle, fill =red, inner sep =1.5pt] at (1.5,0.866){};
            \node(d)[circle, fill, inner sep =1.5pt] at (2,0){};
            \node(e)[circle, fill = red, inner sep =1.5pt] at (1.5,-0.866){};
            \node(f)[circle, fill, inner sep =1.5pt] at (0.5,-0.866){};

            \draw[line width = 0.5mm](a)--(b);
            \draw[line width = 0.5mm](b)--(c);
            \draw[dashed](c)--(d);
            \draw[dashed](d)--(e);
            \draw[dashed](e)--(f);
            \draw[dashed](f)--(a);

            \draw[->][line width = 0.5mm] (c)--(2,1.73);
            \draw[->][line width =0.5mm] (a)--(-1,0);
            \draw[->][line width = 0.5mm] (e)--(2,-1.73);
            \draw[->][line width = 0.5mm] (e)--(1,-1.73);

            \node(a)[circle, fill, inner sep =1.5pt] at (3,0){};
            \node(b)[circle, fill, inner sep =1.5pt] at (3.5,0.866){};
            \node(c)[circle, fill=red, inner sep =1.5pt] at (4.5,0.866){};
            \node(d)[circle, fill, inner sep =1.5pt] at (5,0){};
            \node(e)[circle, fill=red, inner sep =1.5pt] at (4.5,-0.866){};
            \node(f)[circle, fill=red, inner sep =1.5pt] at (3.5,-0.866){};

            \draw[dashed](a)--(b);
            \draw[dashed](b)--(c);
            \draw[line width = 0.5mm](c)--(d);
            \draw[line width = 0.5mm](d)--(e);
            \draw[line width = 0.5mm](e)--(f);
            \draw[dashed](f)--(a);

            \draw[->][line width = 0.5mm] (c)--(5,1.73);
            \draw[->][line width = 0.5mm] (e)--(5,-1.73);
            \draw[->][line width = 0.5mm] (f)--(3,-1.73);
            \end{tikzpicture}
    \caption[Departure vertices]{For a particular face, the departure vertices of that face are the vertices included in the subgraph with at least one adjacent edge not incident to the face included in the subgraph. The examples of departure vertices are highlighted in red, where the dashed lines are edges of the graph which are not included in the subgraph.}
    \label{fig:departureVertex}
\end{figure}
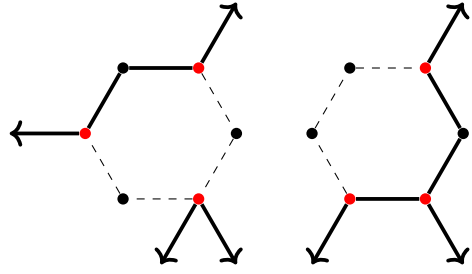

\begin{definition}[Singe Face Rewiring (SFR)]
    For a planar graph $G$ and subgraphs $s,r\in \mathcal{SC}_G$, a Single Face Rewiring (SFR) from $s$ to $r$ is a resampling of edges from $G$ so that for a single face $f$ in $G$ the following are satisfied
    \begin{enumerate}
        \item \textbf{Locality} Any edge not incident to the face $f$ which was in $s$ is in $r$ and any edge not incident to $f$ which was not in $s$ is not in $r$.
        \item \textbf{Closure} Edges incident to $f$ cannot be added or removed from $s$ in such a way that would create a bridge in $r$. 
        \item \textbf{Connectivity} If two departure vertices of $f$ were connected by a path along $f$ in $s$ bur are not connected by a path along $f$ in $r$ then they must be entirely disconnected in $r$.
    \end{enumerate}
\end{definition}

The definition of a single face rewiring seems complicated, but they are indeed easy to imagine. In simpler words, a single face rewiring is a perturbation of the edges in a subgraph $s$ around a single face of the graph $G$ that maintains a certain connectivity condition. There are examples of single face rewirings in figures \ref{fig:SFR1} and \ref{fig:SFR2}

\begin{figure}
    \centering
         \begin{tikzpicture}[scale = 0.85]
            \node(a)[circle, fill, inner sep = 1.5pt] at (0,0){};
            \node(b) [circle,fill, inner sep = 1.5pt]at (1,1.1){};
            \node(c) [circle, fill,inner sep = 1.5pt]at (1.2,2.2){};
            \node(d) [circle, fill,inner sep = 1.5pt]at (2,2.8){};
            \node(e) [circle, fill,inner sep = 1.5pt]at (3.1,-1){};
            \node(f) [circle, fill,inner sep = 1.5pt]at (3.1,1){};
            \node(g) [ fill,inner sep = 3pt]at (3.1,2.4){};
            \node(h) [fill,inner sep = 3pt]at (4.2,0.1){};
            \node(i) [circle, fill,inner sep = 1.5pt]at (4.1,3.4){};
            \node(j) [circle,fill, inner sep = 1.5pt]at (5.2,-1.1){};
            \node(k) [circle, fill,inner sep = 1.5pt]at (5,0.1){};
            \node(l) [circle,fill, inner sep = 1.5pt]at (5.4,3.2){};
            \node(m) [circle, fill,inner sep = 1.5pt]at (6,4){};
            \node(n) [circle, fill,inner sep = 1.5pt]at (6.2,0.2){};
            \node(o) [fill,inner sep = 3pt]at (5.7,2.1){};
            \node(p) [circle, fill,inner sep = 1.5pt]at (6,3.2){};

            \draw[dashed,line width = 0.2mm](a)--(b);
            \draw[dashed,line width = 0.2mm](a)--(h);
            \draw[dashed,line width = 0.2mm](a)--(e);
            \draw[dashed,line width = 0.2mm](b)--(c);
            \draw[dashed,line width = 0.2mm](b)--(f);
            \draw[dashed,line width = 0.2mm](b)--(h);
            \draw[dashed,line width = 0.2mm](c)--(d);
            \draw[dashed,line width = 0.2mm](c)--(f);
            \draw[dashed,line width = 0.2mm](c)--(g);
            \draw[dashed,line width = 0.2mm](d)--(g);
            \draw[dashed,line width = 0.2mm](d)--(i);
            \draw[dashed,line width = 0.2mm](e)--(h);
            \draw[dashed,line width = 0.2mm](e)--(j);
            \draw[dashed,line width = 0.2mm](f)--(h);
            \draw[dashed,line width = 0.2mm](g)--(f);
            \draw[dashed,line width = 0.2mm](g)--(i);
            \draw[dashed,line width = 0.2mm](h)--(k);
            \draw[dashed,line width = 0.2mm](h)--(o);
            \draw[dashed,line width = 0.2mm](i)--(l);
            \draw[dashed,line width = 0.2mm](i)--(m);
            \draw[dashed,line width = 0.2mm](j)--(k);
            \draw[dashed,line width = 0.2mm](j)--(n);
            \draw[dashed,line width = 0.2mm](k)--(n);
            \draw[dashed,line width = 0.2mm](k)--(o);
            \draw[dashed,line width = 0.2mm](l)--(o);
            \draw[dashed,line width = 0.2mm](l)--(p);
            \draw[dashed,line width = 0.2mm](l)--(m);
            \draw[dashed,line width = 0.2mm](m)--(p);
            \draw[dashed,line width = 0.2mm](n)--(o);
            \draw[dashed,line width = 0.2mm](o)--(p);

            \node(lbl) at (1,3) {$s_1$};
            
            \draw[thick,line width = 0.5mm](a)--(b)--(c)--(d)--(g)--(f)--(h)--(e)--(a);
            \draw[thick,line width = 0.5mm](e)--(j)--(n)--(o)--(h)--(e);

            \filldraw[fill=gray, opacity =0.2] (3.1,1)--(3.1,2.4)--(4.1,3.4)--(5.4,3.2)--(5.7,2.1)--(4.2,0.1)--cycle;

            \node(lbl) at (7,1){$\leftrightarrow$};
            
            \end{tikzpicture}
            \begin{tikzpicture}[scale = 0.85]
            \node(a)[circle, fill, inner sep = 1.5pt] at (0,0){};
            \node(b) [circle,fill, inner sep = 1.5pt]at (1,1.1){};
            \node(c) [circle, fill,inner sep = 1.5pt]at (1.2,2.2){};
            \node(d) [circle, fill,inner sep = 1.5pt]at (2,2.8){};
            \node(e) [circle, fill,inner sep = 1.5pt]at (3.1,-1){};
            \node(f) [circle, fill,inner sep = 1.5pt]at (3.1,1){};
            \node(g) [ fill,inner sep = 3pt]at (3.1,2.4){};
            \node(h) [fill,inner sep = 3pt]at (4.2,0.1){};
            \node(i) [circle, fill,inner sep = 1.5pt]at (4.1,3.4){};
            \node(j) [circle,fill, inner sep = 1.5pt]at (5.2,-1.1){};
            \node(k) [circle, fill,inner sep = 1.5pt]at (5,0.1){};
            \node(l) [circle,fill, inner sep = 1.5pt]at (5.4,3.2){};
            \node(m) [circle, fill,inner sep = 1.5pt]at (6,4){};
            \node(n) [circle, fill,inner sep = 1.5pt]at (6.2,0.2){};
            \node(o) [fill,inner sep = 3pt]at (5.7,2.1){};
            \node(p) [circle, fill,inner sep = 1.5pt]at (6,3.2){};
            \draw[dashed,line width = 0.2mm](a)--(b);
            \draw[dashed,line width = 0.2mm](a)--(h);
            \draw[dashed,line width = 0.2mm](a)--(e);
            \draw[dashed,line width = 0.2mm](b)--(c);
            \draw[dashed,line width = 0.2mm](b)--(f);
            \draw[dashed,line width = 0.2mm](b)--(h);
            \draw[dashed,line width = 0.2mm](c)--(d);
            \draw[dashed,line width = 0.2mm](c)--(f);
            \draw[dashed,line width = 0.2mm](c)--(g);
            \draw[dashed,line width = 0.2mm](d)--(g);
            \draw[dashed,line width = 0.2mm](d)--(i);
            \draw[dashed,line width = 0.2mm](e)--(h);
            \draw[dashed,line width = 0.2mm](e)--(j);
            \draw[dashed,line width = 0.2mm](f)--(h);
            \draw[dashed,line width = 0.2mm](g)--(f);
            \draw[dashed,line width = 0.2mm](g)--(i);
            \draw[dashed,line width = 0.2mm](h)--(k);
            \draw[dashed,line width = 0.2mm](h)--(o);
            \draw[dashed,line width = 0.2mm](i)--(l);
            \draw[dashed,line width = 0.2mm](i)--(m);
            \draw[dashed,line width = 0.2mm](j)--(k);
            \draw[dashed,line width = 0.2mm](j)--(n);
            \draw[dashed,line width = 0.2mm](k)--(n);
            \draw[dashed,line width = 0.2mm](k)--(o);
            \draw[dashed,line width = 0.2mm](l)--(o);
            \draw[dashed,line width = 0.2mm](l)--(p);
            \draw[dashed,line width = 0.2mm](l)--(m);
            \draw[dashed,line width = 0.2mm](m)--(p);
            \draw[dashed,line width = 0.2mm](n)--(o);
            \draw[dashed,line width = 0.2mm](o)--(p);
            
            \node(lbl) at (1,3) {$s_2$};
           
            \filldraw[fill=gray, opacity =0.2] (3.1,1)--(3.1,2.4)--(4.1,3.4)--(5.4,3.2)--(5.7,2.1)--(4.2,0.1)--cycle;

            \draw[thick,line width = 0.5mm](a)--(b)--(c)--(d)--(g)--(i)--(l)--(o)--(h)--(e)--(a);
            \draw[thick,line width =0.5mm](e)--(j)--(n)--(o)--(h)--(e);
            
            \end{tikzpicture}
    \caption{A reversible single face rewiring for a face $f$ (colored in gray) of the planar graph $G$. All edges of $G$ are shown. If they are included in the subgraph $s_1$ on the left or $s_2$ on the right, they are filled in. If they are not included, they are dashed. Departure vertices of the face $f$ are shown as black squares. The resampling of edges from $s_1$ to $s_2$ satisfies Locality, Closure, and Connectivity, and thus it is a single face rewiring. The same is true about the resampling from $s_2$ to $s_1$.}
    \label{fig:SFR1}
\end{figure}
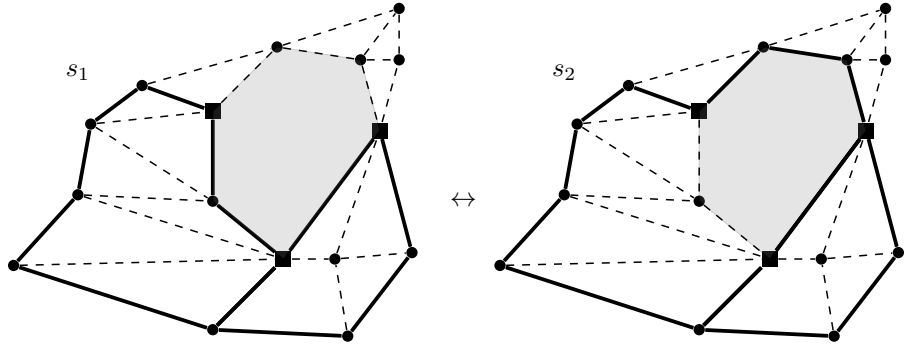
\begin{figure}
    \centering
    \begin{tikzpicture}[scale =0.85]
            \node(a)[circle, fill, inner sep = 1.5pt] at (0,0){};
            \node(b) [circle,fill, inner sep = 1.5pt]at (1,1.1){};
            \node(c) [circle, fill,inner sep = 1.5pt]at (1.2,2.2){};
            \node(d) [circle, fill,inner sep = 1.5pt]at (2,2.8){};
            \node(e) [circle, fill,inner sep = 1.5pt]at (3.1,-1){};
            \node(f) [fill=orange,inner sep = 3pt]at (3.1,1){};
            \node(g) [circle, fill,inner sep = 1.5pt]at (3.1,2.4){};
            \node(h) [ fill=blue,inner sep = 3pt]at (4.2,0.1){};
            \node(i) [circle, fill,inner sep = 1.5pt]at (4.1,3.4){};
            \node(j) [circle,fill, inner sep = 1.5pt]at (5.2,-1.1){};
            \node(k) [circle, fill,inner sep = 1.5pt]at (5,0.1){};
            \node(l) [fill=orange, inner sep = 3pt]at (5.4,3.2){};
            \node(m) [circle, fill,inner sep = 1.5pt]at (6,4){};
            \node(n) [circle, fill,inner sep = 1.5pt]at (6.2,0.2){};
            \node(o) [fill=blue,inner sep = 3pt]at (5.7,2.1){};
            \node(p) [circle, fill,inner sep = 1.5pt]at (6,3.2){};

            \draw[dashed,line width = 0.2mm](a)--(b);
            \draw[dashed,line width = 0.2mm](a)--(h);
            \draw[dashed,line width = 0.2mm](a)--(e);
            \draw[dashed,line width = 0.2mm](b)--(c);
            \draw[dashed,line width = 0.2mm](b)--(f);
            \draw[dashed,line width = 0.2mm](b)--(h);
            \draw[dashed,line width = 0.2mm](c)--(d);
            \draw[dashed,line width = 0.2mm](c)--(f);
            \draw[dashed,line width = 0.2mm](c)--(g);
            \draw[dashed,line width = 0.2mm](d)--(g);
            \draw[dashed,line width = 0.2mm](d)--(i);
            \draw[dashed,line width = 0.2mm](e)--(h);
            \draw[dashed,line width = 0.2mm](e)--(j);
            \draw[dashed,line width = 0.2mm](f)--(h);
            \draw[dashed,line width = 0.2mm](g)--(f);
            \draw[dashed,line width = 0.2mm](g)--(i);
            \draw[dashed,line width = 0.2mm](h)--(k);
            \draw[dashed,line width = 0.2mm](h)--(o);
            \draw[dashed,line width = 0.2mm](i)--(l);
            \draw[dashed,line width = 0.2mm](i)--(m);
            \draw[dashed,line width = 0.2mm](j)--(k);
            \draw[dashed,line width = 0.2mm](j)--(n);
            \draw[dashed,line width = 0.2mm](k)--(n);
            \draw[dashed,line width = 0.2mm](k)--(o);
            \draw[dashed,line width = 0.2mm](l)--(o);
            \draw[dashed,line width = 0.2mm](l)--(p);
            \draw[dashed,line width = 0.2mm](l)--(m);
            \draw[dashed,line width = 0.2mm](m)--(p);
            \draw[dashed,line width = 0.2mm](n)--(o);
            \draw[dashed,line width = 0.2mm](o)--(p);

            \draw[thick,line width = 0.5mm](a)--(b)--(f)--(g)--(i)--(l)--(m)--(p)--(o)--(h)--(a);

            \filldraw[fill=gray, opacity =0.2] (3.1,1)--(3.1,2.4)--(4.1,3.4)--(5.4,3.2)--(5.7,2.1)--(4.2,0.1)--cycle;

            \node(lbl) at (7,1){$\to$};
            
            \end{tikzpicture}
                \begin{tikzpicture}[scale =0.85]
            \node(a)[circle, fill, inner sep = 1.5pt] at (0,0){};
            \node(b) [circle,fill, inner sep = 1.5pt]at (1,1.1){};
            \node(c) [circle, fill,inner sep = 1.5pt]at (1.2,2.2){};
            \node(d) [circle, fill,inner sep = 1.5pt]at (2,2.8){};
            \node(e) [circle, fill,inner sep = 1.5pt]at (3.1,-1){};
            \node(f) [fill=orange,inner sep = 3pt]at (3.1,1){};
            \node(g) [circle, fill,inner sep = 1.5pt]at (3.1,2.4){};
            \node(h) [ fill=blue,inner sep = 3pt]at (4.2,0.1){};
            \node(i) [circle, fill,inner sep = 1.5pt]at (4.1,3.4){};
            \node(j) [circle,fill, inner sep = 1.5pt]at (5.2,-1.1){};
            \node(k) [circle, fill,inner sep = 1.5pt]at (5,0.1){};
            \node(l) [fill=orange, inner sep = 3pt]at (5.4,3.2){};
            \node(m) [circle, fill,inner sep = 1.5pt]at (6,4){};
            \node(n) [circle, fill,inner sep = 1.5pt]at (6.2,0.2){};
            \node(o) [fill=blue,inner sep = 3pt]at (5.7,2.1){};
            \node(p) [circle, fill,inner sep = 1.5pt]at (6,3.2){};

            \draw[dashed,line width = 0.2mm](a)--(b);
            \draw[dashed,line width = 0.2mm](a)--(h);
            \draw[dashed,line width = 0.2mm](a)--(e);
            \draw[dashed,line width = 0.2mm](b)--(c);
            \draw[dashed,line width = 0.2mm](b)--(f);
            \draw[dashed,line width = 0.2mm](b)--(h);
            \draw[dashed,line width = 0.2mm](c)--(d);
            \draw[dashed,line width = 0.2mm](c)--(f);
            \draw[dashed,line width = 0.2mm](c)--(g);
            \draw[dashed,line width = 0.2mm](d)--(g);
            \draw[dashed,line width = 0.2mm](d)--(i);
            \draw[dashed,line width = 0.2mm](e)--(h);
            \draw[dashed,line width = 0.2mm](e)--(j);
            \draw[dashed,line width = 0.2mm](f)--(h);
            \draw[dashed,line width = 0.2mm](g)--(f);
            \draw[dashed,line width = 0.2mm](g)--(i);
            \draw[dashed,line width = 0.2mm](h)--(k);
            \draw[dashed,line width = 0.2mm](h)--(o);
            \draw[dashed,line width = 0.2mm](i)--(l);
            \draw[dashed,line width = 0.2mm](i)--(m);
            \draw[dashed,line width = 0.2mm](j)--(k);
            \draw[dashed,line width = 0.2mm](j)--(n);
            \draw[dashed,line width = 0.2mm](k)--(n);
            \draw[dashed,line width = 0.2mm](k)--(o);
            \draw[dashed,line width = 0.2mm](l)--(o);
            \draw[dashed,line width = 0.2mm](l)--(p);
            \draw[dashed,line width = 0.2mm](l)--(m);
            \draw[dashed,line width = 0.2mm](m)--(p);
            \draw[dashed,line width = 0.2mm](n)--(o);
            \draw[dashed,line width = 0.2mm](o)--(p);

            \filldraw[fill=gray, opacity =0.2] (3.1,1)--(3.1,2.4)--(4.1,3.4)--(5.4,3.2)--(5.7,2.1)--(4.2,0.1)--cycle;

            \draw[thick,line width = 0.5mm](a)--(b)--(f)--(h)--(a);
            \draw[thick,line width =0.5mm](l)--(m)--(p)--(o)--(l);
            
            \end{tikzpicture}
    \caption{An irreversible single face rewiring for a face $f$ (colored in gray) of the planar graph $G$. All edges of $G$ are shown. If they are included in the subgraph $s_1$ on the left or $s_2$ on the right, they are filled in. If they are not included, they are dashed. Departure vertices of the face $f$ are shown as black squares. The resampling of edges from $s_1$ to $s_2$ satisfies Locality, Closure, and Connectivity, and thus it is a single face rewiring. Note that the orange departure vertices become disconnected by a path along the face $f$ and become entirely disconnected in $s_2$. The same is \textit{not} true of the reverse direction. To resample edges from $s_2$ to $s_1$, either pair of blue and orange departure vertices becomes disconnected by a path along the face $f$, but both are still connected in $s_1$.}
    \label{fig:SFR2}
\end{figure}
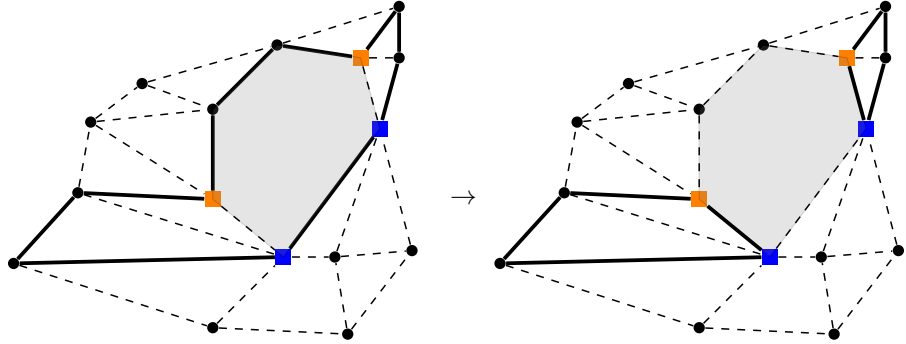

It is important to note that the connectivity condition means that not every single face rewiring is reversible, and so, if we imagine every element of $\mathcal{SC}_G$ as a vertex in a graph and put an edge from one vertex to another if we can get from one to the other by a single face rewiring, the result is a \textit{directed} graph. However, this directed graph is enough to define a quasimetric space. 

\begin{definition}[Bridge-free subgraph space $(\mathcal{SC}_G,d_{\mathcal{SC}}(\cdot,\cdot))$]
The bridge-free subgraph space is the set of all bridge-free subgraphs of $G$ equipped with the metric induced by the directed adjacency of single face rewirings. In particular $d_{\mathcal{SC}}(s_1,s_2)$ is the minimum number of SFRs required to get from $s_1$ to $s_2$. 
\end{definition}
This is a quasimetric exactly because the connectivity condition prevents every adjacency from being symmetric. However, the quasimetric is sufficient for our purposes.

\subsection{Relation between the strategy space and the subgraph space}

For readers unfamiliar with the graphical dual, the standard results which are used in this result are described in \cite{NishizekiPlanarGraphs}. The most important point is that for a planar graph $G$ (and a particular planar embedding), each face of $G$ corresponds to exactly 1 vertex in $G^*$, and each edge in $G$ which separates two faces corresponds to the unique edge in $G^*$ which connects the two vertices corresponding to those faces. This implies that there is a bijection between edges in $G$ and edges in $G^*$, a bijection between faces in $G$ and vertices in $G^*$, and a bijection between vertices of $G$ and faces of $G^*$ for any particular planar embedding of $G$. To ensure that the bijective relationship still holds even without determining a planar embedding, we need only require that $G$ is 3-edge connected and thus has a unique dual. We will assume, for the following, that a particular planar embedding of $G$ is determined, but that assumption can be relaxed if the graph is 3-edge connected. 

For planar graphs, there is a standard relationship between a vertex partition of a planar graph $G$ and a bridge-free subgraph of the dual $G*$. This relationship is easily described through standard facts about the graphical dual. The relationship between edges of $G$ and edges of $G^*$ means that we can relate the cut set of a partition to a subset of edges in $G^*$. Let $\Psi$ be a relation between vertex partitions of $G$ and subgraphs of $G^*$ so that $\Psi P$ is the subgraph which is spanned by the edges of $G^*$ corresponding to the edges of $G$ included in the cut set of $P$. This is the standard relation to show the isomorphism between the cut space of $G$ and the cycle space of $G^*$. Indeed, the set of partitions is a generalization of the cutspace (which ruins the algebraic properties of the space), and the subgraph space is similarly a generalization of the cyclespace without the algebraic properties. 

This relation $\Psi$, together with the function $\Phi$, will make a quasiisometry between $\mathcal{A}_G$ and $\mathcal{SC}_{G^*}$. The main result of this section is to use that quasiisometry to show that subgraphs that locally minimize their size in the subgraph space are associated with Nash equilibria (Theorem \ref{thm:minimalSubgraphs}). The first step is to show that $\Psi$ is a surjection and determine how to restrict the domain so that it is a bijection.

\begin{lemma}\label{lem:surjectivePsi}
    $\Psi:\mathcal{Q}_G\to \mathcal{SC}_{G^*}$ is a surjection
\end{lemma}
\begin{proof}
    Consider any vertex partition $P\in\mathcal{Q}_G$ and write it as $\{P_1,P_2,...,P_n\}.$ Let $P_{-i}$ be the partition made from $\{P_i, P_i^c\}$. This is a partition into two parts, and so the cutset of $P_i$ is a cut of the graph $G$. It is a standard result from planar graph theory \citep{FouldsGraphTheory,NishizekiPlanarGraphs,Gross2005} that a cut of a planar graph corresponds to an element of the cycle space because the cutspace and the cyclespace are isomorphic vector spaces. The isomorphism between the two spaces is defined exactly like $\Psi$, so that for every edge $e$ in the cut of $G$, its dual counterpart, $e^*$, is part of the cycle in $G^*$.

    Because of this correspondence, it is obvious that $\Psi P_{-i}$ is in $\mathcal{SC}_{G^*}$. Notice now that 
    \[\Psi P=\bigcup_{i=1}^n \Psi P_{-i}\]
    because the cutset of $P$ is the union of cutsets of $P_{-i}$. Thus $\Psi P$ is the union of cycles and thus necessarily bridge-free. Thus we have shown that $\Psi:\mathcal{Q}_G\to \mathcal{SC}_{G^*}$ is well defined.

    To show that it is a surjection, consider any bridge-free subgraph $s$ of $G^*$. For a planar embedding of $G$ and a corresponding planar embedding of $G^*$, we can enumerate the faces of $s$ and write the set of face cycles of $s$ as $F_s:=\{c_1,...,c_m\}$. Each cycle in $F_s$ is a Jordan curve on the plane that separates a face of $s$ from all the other faces of $s$. Each face of $s$ is a collection of faces of $G^*$, and each face in $G^*$ corresponds to a vertex in $G$. Thus, each face in $s$ corresponds to a collection of vertices in $G$. Let $P_i$ be the collection of vertices in $G$ which correspond to the face surrounded by $c_i$.

    We must note two things about these cycles. First, because $s$ is bridge-free, every face is indeed surrounded by a simple cycle. This is to say that there is no edge in $s$ which is incident to the same face on both sides. The other is that each $c_i$ is ``minimal". This means that if we call the face surrounded by $c_i$ the ``inside" of $c_i$, then it is the smallest cycle which separates any point on the inside from every point on the outside. If this were not the case, then $c_i$ would surround at least two faces of $s$.

    Using these face cycles, we are left with a set of sets of vertices $P:=\{P_1,...,P_m\}$. Note that every vertex is included in a set and no vertex is included in more than one set, so $P$ is a vertex partition. Now we must show that $\Psi P=s$. Observe that if an edge $e^*\in \Psi P_{-i}\subset \Psi P$, then the corresponding edge $e$ connects a vertex in $P_i$ and a vertex not in $P_i$. These two vertices correspond to faces of $G^*$ which were not included in the same face of $s$. The faces are adjacent in $G^*$ because they share an incident edge $e^*$, so that edge $e^*$ must be included in $s$ to ensure the two faces of $G^*$ are not included in the same face of $s$. What we have just shown is that $\Psi P\subset s$. By construction, every edge of $s$ separates two faces of $G^*$ which are not in the same face of $s$. Thus, the corresponding vertices are not in the same part of the partition $P$. This means that every edge in $s$ corresponds to an edge which is in the cutset of the partition $P$ and so $s\subset \Psi P$.
    
     Thus, we have shown that for every $s\in \mathcal{SC}_{G^*}$ there must be a $P\in \mathcal{Q}_G$ so that $\Psi P=s$.
\end{proof}

\begin{lemma}\label{lem:bijectivePsi}
    Let $\mathcal{QC}_G$ be the set of partitions of $G$ whose parts are connected. $\Psi:\mathcal{QC}_G\to\mathcal{SC}_{G^*}$ is a bijection. 
\end{lemma}
\begin{proof}
    First, we will use the construction from Lemma \ref{lem:surjectivePsi} to show that there is a partition in $\mathcal{QC}_G$ in the preimage of any $s\in \mathcal{SC}_{G^*}$. 
    Let $sp(P_i)$ be the subgraph of $G$ with the vertices from $P_i$ and the edges which connects two vertices in $P_i$.

    consider a minimal cycle $c_1$ in $s\subset G^*$ that encloses a set of vertices $Q$ in $G$ so that $sp(Q)$ at least two connected components. Take one component, $Q_1$, and consider the vertex partition $\{Q_1,Q_1^c\}$. The cut that forms this partition is mapped to a cycle $c_2$ in $G^*$ by $\Psi$. Note that every edge of $c_2$ is also included in $c_1$, and $c_2$ is a cycle, so $c_2$ separates a point from inside $c_1$ from every point outside $c_1$, and thus  $c_1$ was not minimal. Thus, we have shown by contradiction that the sets $P_i$ constructed in Lemma \ref{lem:surjectivePsi} are connected. Therefore, for every element $\mathcal{SC}_{G^*}$ there is a partition in $\mathcal{QC}_G$ in its preimage under $\Psi$ and so $\Psi$ is still a surjection with the restricted domain.

    To show that it is injective, consider two partitions $P,Q\in \mathcal{QC}_{G}$. First, note that a partition with connected parts is uniquely determined by its cutset. If $P\neq Q$, then without loss of generality, there is a pair of vertices, $v$ and $w$, which are in the same part of $P$ but not in the same part of $Q$. In this way, suppose $v,w\in P_1$. Thus, there is a path from $v$ to $w$ in $P_1$ (i.e., every edge of the path is in the same part of $P$). There is at least one edge in that path that connects two vertices in different parts of $Q$, so there is an edge which was not in the cutset of $P$ but is in the cutset of $Q$. Because of the bijective mapping of edges in $G$ to edges in $G^*$, it is clear that if $\Psi P=\Psi Q$ then the cutsets of $P$ and $Q$ are identical and thus $P$ and $Q$ are the same partition if they are both in $\mathcal{QC}_G$. This proves that $\Psi:\mathcal{QC}_{G}\to \mathcal{SC}_{G^*}$ is bijective. 
\end{proof}

    Having shown this, we know that we have a surjection from the strategy space to the subgraph space, $\Psi\circ\Phi$. Now we can use this function to translate the potential game results from Corollary \ref{cor:mincutPotential} to the subgraph space
\begin{lemma}\label{lem:Potentialissize}
    Let $\mathcal{V}(s)=\frac{-1}{2}size(s)$ where $size(s)$ is the number of edges in the subgraph $s$. With this construction, $\mathcal{W}(u)=\mathcal{V}(\Psi\circ\Phi u)$ for all $u\in \mathcal{A}$ 
\end{lemma}
\begin{proof}
    For a partition $P\in \mathcal{Q}_{G}$, $E(P)= 2m -\sum_{v,w}W_{v,w}\delta(P_v,P_w)$, Where $\delta(P_v,P_w)$ returns a 1 if $v$ and $w$ are in the same part of the partition and 0 otherwise. Naturally, $E(P)$ enumerates the cutset of the partition $P$.  Because every cut edge in $P$ is an edge in the subgraph $\Psi P\subset G^*$, it is clear that $E(P)=size(\Phi P)$ for all $P\in \mathcal{Q}_G$. Therefore, for any $u\in \mathcal{A}$, \[\mathcal{V}(\Psi\circ\Phi u)=\frac{-1}{2}size(\Psi\circ\Phi u)=\frac{-1}{2}E(\Phi u)=\mathcal{W}(u)\] 
\end{proof}

This result has two important corollaries.

\begin{corollary}\label{cor:equalpotential}
    If $u,u'\in \mathcal{A}$ and $\Psi\circ\Phi u = \Psi\circ\Phi u' =s$ then $\mathcal{W}(u)=\mathcal{W}(u')$
\end{corollary}
\begin{proof}
    This follows directly from the direct relationship between size and potential in Lemma \ref{lem:Potentialissize}.
\end{proof}

\begin{corollary}\label{cor:minimizingsize}
    If $G$ is a planar graph (with a determined planar embedding) and $G^*$ is its dual, and if $u\in \mathcal{A}$ is a Nash equilibrium, then $size(\Psi\circ\Phi u)\leq size(\Psi\circ\Phi u')$ for all $u'\in \Gamma_{\mathcal{A}}(u).$ 
\end{corollary}
\begin{proof}
    This follows directly from Corollary \ref{cor:mincutPotential} and Lemma \ref{lem:Potentialissize}.
\end{proof}

This lemma and its associated corollaries tell us that maximizing game potential in the strategy space is the same as minimizing the size of the associated subgraph. In order to reach the main result, the only thing we need to do is show that our ideas of locality can be translated from strategy space to subgraph space. In order to do that, we establish that the function $\Psi\circ\Phi:\mathcal{A}_G\to \mathcal{SC}_{G^*}$ is a quasiisometry. 

    \begin{lemma}\label{lem:MetricSimilarity}
        If $u, u'\in \mathcal{A}_G$ and $\Phi u \in \mathcal{QC}_{G}$ satisfy $d_\mathcal{A}(u,u')=1$ in $\mathcal{A}_G$, then $d_{\mathcal{SC}}(\Psi\circ\Phi u,\Psi\circ\Phi u')=1$ in $\mathcal{SC}_{G^*}$.
    \end{lemma}

    \begin{proof}
        Given two strategy profiles $u$ and $u'$ that are adjacent in $\mathcal{A}_G$, we seek to show that we can get from $s:=\Psi\circ\Phi u$ to $s':=\Psi\circ\Phi u'$ using an SFR. It will be sufficient to show that starting from $s$ we can resample the edges and be left with $s'$ and satisfy the three conditions: \textbf{Locality, Closure} and \textbf{Connectivity}.

        \textbf{Locality}: Let $v$ be the vertex that changes its strategy in the single strategic change from $u$ to $u'$. Let $f$ be the face in $G^*$ which corresponds to $v$. For notational ease, let $Cs(P)$ be the cutset of the partition $P$. Select an edge $e$ in $G$ which is not adjacent to $v$ and call the corresponding edge in $G^*$, $e^*$.
        If $e\in Cs(\Phi u)$ then, because neither vertex adjacent to $e$ changed its strategy, $e\in sC(\Phi u')$. Likewise $e\notin Cs(\Phi u)\iff e\notin Cs(\Phi u')$. By the definition of $\Psi$, $e^*\in \Psi P\iff e\in Cs(P)$ so we get the two logical strings
        \[e^*\in s\iff e\in Cs(\Phi u)\iff e\in Cs(\Phi u') \iff e^*\in s'\]
        \[e^*\notin s\iff e\notin Cs(\Phi u)\iff e\notin Cs(\Phi u') \iff e^*\notin s'\]
        for every $e^\star$ not incident to the face $f$ (meaning $e$ is not adjacent to vertex $v$). This proves that the resampling from $s$ to $s'$ satisfies the locality condition

        \textbf{Closure}: This condition is satisfied directly by lemma \ref{lem:surjectivePsi}. The image of $\Psi$ is always bridge-free and thus, regardless of the adjacency between $u$ and $u'$, the resampling cannot have added a bridge, and so the resampling from $s$ to $s'$ satisfies the closure condition. 

        \textbf{Connectivity}: This condition is the most complicated, but it can be accomplished by taking advantage of the assumption that $\Phi u\in \mathcal{QC}_G$. Assume that two vertices $a$ and $b$ in $G^*$ are connected by a path incident to $f$ in $s$ but not in $s'$. This means that there are at least two non-adjacent edges incident to $f$ in $G^*$ that are not present in $s'$. More specifically, there are two paths from $a$ to $b$ in $G^*$ incident to $f$, and there is at least one edge incident to $f$ that is not present in $s'$ from each of those two paths. 
        
        This means that for any strategy profile in the preimage, $\Phi^{-1}\circ\Psi^{-1}(s')$. there are two vertices $w^+,w^-$ adjacent to $v$ that are using the same strategy as $v$ in $u'$. Because $w^-$ and $w^+$ did not change their strategy from $u$ to $u'$, it must also be true that they were using the same strategy in $u$. Because we assumed that $\Psi u\in \mathcal{QC}_{G}$, we know that all parts of the partition are connected and thus $w^-$ and $w^+$ are connected by a path of vertices using the same strategy. When $v$ changes its strategy to take on the same strategy as $w^-$ and $w^+$ in $u'$, there is now necessarily a cycle starting and ending at $v$ of vertices using the same strategy in $u'$. 

        From here, we rely on the planar properties of $G$ and $G^*$. The cycle of vertices in the same strategic community starting and ending at $v$ forms a Jordan curve in the plane. In the dual, that same Jordan curve starts and ends in the face $f$ and intersects every edge of $G^*$ corresponding to the edges of $G$, which defines the curve. The Jordan curve theorem tells us that the plane is separated into two parts: an ``inside" and an ``outside." 

        Observe that on any path from $a$ to $b$ in $G^*$ incident to $f$, one of the edges must be absent from $s'$ and correspond to an edge of $G$ that defines the Jordan curve. This means that the Jordan curve intersects both paths from $a$ to $b$ incident to $f$. Importantly, it must intersect both paths exactly once because the cycle in $G$ included only two edges incident to $v$, and so the Jordan curve can intersect only the two corresponding edges of $G^*.$ This implies that $a$ and $b$ cannot both be ``inside" and the cannot both be ``outside." One must be inside, and one must be outside. 

        By construction, every edge of $G^*$ that the Jordan curve intersects is \textit{not} included in $s'$. Because any path from $a$ to $b$ in $G$ must cross the Jordan curve, it is certain that $a$ and $b$ are disconnected in $s'$. Therefore, the resampling from $s$ to $s'$ satisfies the connectivity condition. 

        This means that if $u$ and $u'$ are a single strategic change away from one another, then $\Psi\circ\Phi u$ can be made into $\Psi \circ \Phi u'$ by a resampling of the edges which satisfies all three conditions of the SFR. Therefore $d_{\mathcal{SC}}(\Psi\circ\Phi u,\Psi\circ\Phi u')=1$ in $\mathcal{SG}_{G^*}$. 
        \end{proof}

    With this lemma, we are finally able to prove the main result.
    \begin{theorem}\label{thm:minimalSubgraphs}
    If $s\in \mathcal{SC}_{G^*}$ satisfies \[size(s)\leq size(r)\quad \forall r\in \Gamma_{\mathcal{SC}}(s)\]
    where $\Gamma_{\mathcal{SC}}(s):=\{t\in \mathcal{SC}_{G^*};d_{\mathcal{SC}}(s,t)\leq 1\}$,
    then there is a $u\in \Phi^{-1}\Psi^{-1}s$ so that $u$ is a Nash equilibrium. Moreover, this $u$ will satisfy $\Phi u\in \mathcal{QC}_G$ and so $u=\Phi^{-1}\circ\Psi^{-1}s$ under the domain restriction. 
    \end{theorem}

    \begin{proof}
    By Lemma \ref{lem:MetricSimilarity} we know that if $u$ has connected strategic communities (i.e., $u\in \Phi^{-1}\mathcal{QC}_G$) and  $d_{\mathcal{A}}(u,u')\leq 1$ in $\mathcal{A}$ then $d_{\mathcal{SC}}(\Psi\circ \Phi u,\Psi\circ\Phi u')\leq 1$. This means that for $u\in \Phi^{-1}\mathcal{QC}_G$, $\Psi\circ\Phi(\Gamma_\mathcal{A}(u))\subseteq\Gamma_{\mathcal{SC}}(\Psi\circ \Phi u)$.  Equivalently $\Gamma_\mathcal{A}(u)\subset \Phi^{-1}\circ\Psi^{-1}(\Gamma_{\mathcal{SC}}(\Psi\circ\Phi u)$ for  $u\in\Phi^{-1}\mathcal{QC}_G$.  We can complete the first part of the proof by noting that by lemma \ref{lem:Potentialissize} and corollary \ref{cor:equalpotential} if $size(\Psi\circ\Phi u)\leq size(s)$ for all $s\in \Gamma_{\mathcal{SC}}(\Psi\circ\Phi u)$ then \[\mathcal{W}(u)\geq \mathcal{W}(u')\quad \forall u'\in \Phi^{-1}\circ\Psi^{-1}(\Gamma_{\mathcal{SC}}(\Psi\circ\Phi u)).\] 

    Putting these together this means that if $u\in  \Phi^{-1}\mathcal{QC}_G$ and $size(\Psi\circ\Phi u)\leq size(s)$ for all $s\in \Gamma_{\mathcal{SC}}(\Psi\circ\Phi u)$,  then $\mathcal{W}(u)\geq \mathcal{W}(u')$ for all $u'\in \Gamma_{\mathcal{A}}(u)$. Lastly, recall that by lemma \ref{lem:bijectivePsi} for every $s\in \mathcal{SC}_{G^*}$ there is a unique $u\in \Phi^{-1}\circ\Psi^{-1}s$ such that $\Phi u\in \mathcal{QC}_G$. This implies that if $s$ minimizes size in the neighborhood $\Gamma_{\mathcal{SC}}(s)$ then there is a $u\in \Phi^{-1}\circ\Psi^{-1} s$ with $\Phi u\in \mathcal{QC}_G$ that satisfies $\mathcal{W}(u)\geq \mathcal{W}(u')$ for all $u'\in \Gamma_\mathcal{A}(u)$. Potential game theory implies that such a $u$ is a Nash equilibrium. 
    \end{proof}

    This main result tells us that given a planar embedding of a planar graph $G$, we can search for locally minimal subgraphs of the dual $G^*$ to find Nash equilibria of the coordination game and, because of the bijective relationship between $\mathcal{QC}_G$ and $\mathcal{SC}_{G^*}$ we know exactly how to construct the strategy profile which is the Nash equilibrium. This does two important things. It demonstrates that the equilibrium partition is indeed the local version of the mincut partition, and it gives us information about what these types of partitions look like. In application areas considering pattern formation in space, where a locally driven community detection may be important, this result shows that locally minimizing a global objective defined in the dual setting is an appropriate way to solve the problem. 

    Unfortunately, there are drawbacks. Primarily, the result cannot be made to work in the other direction. Because there exist paths through the subgraph space which do not reflect paths through the strategy space, finding a Nash equilibrium does not give us a way to construct a locally minimal subgraph. The requirement that the domain of $\Psi$ be restricted in order to be a bijection ensures that this direction is not true.


    This lack of a bijection between $\mathcal{A}_G$ and $\mathcal{SC}_{G^*}$ causes problems if we want to consider the system in a dynamic sense as well. If one seeks to understand the game through Synchronous Myopic Best Response, there is no way to capture that behavior in the subgraph space. Performing a Single Face Rewiring at every face synchronously results in the case where one strategic community may become disconnected at one instant and reconnected in a different location in the same instant. In this case, the outcome is not well defined without some prior decisions about how changes to the global properties of the network determine local behavior. 
    
    The mixture of local and global information is the key element that makes this game difficult to study analytically. In the dual, the global information about the connectivity of the system required in order to  ``move around" in the subgraph space prevents us from describing flows in a simple way and makes the prospect of using it to describe the game with synchronous update impossible. 

    There is, however, a reduction of the game, which has been taken by many other scholars studying this system, that allows for the dynamics in the dual to be made entirely clear. We will reduce the game to the case where there are only two strategies. This removes the need for global connectivity information in the dual because, with only two strategies, the cut set of the partition fully describes the strategy profile regardless of whether each strategic community is connected. We describe this more exactly in the following section.
    
\section{Restriction to the two strategy game}\label{sec:mbr}
    In the case where we are restricted only to two strategies, many of the issues with the general setting disappear. This is because of the vector space qualities of vertex partitions into two parts. This will allow us to recreate all of the results of the previous section with more precision. We will redefine a Single Face Rewiring and show that through the relaxed definition it forms a metric on the cycle space of a planar graph. From there, we can use the isomorphism between the cut space of a planar graph and the cycle space of its dual to show a more exact correspondence between locally minimal subgraphs and Nash equilibria as well as define a dynamic graph process, entirely in the subgraph space, which exactly recapitulates the Myopic Best Response process in the strategy space. 

    As before, we start by defining the strategy space, the reduced strategy space, and the subgraph space. The strategy space with only two strategies is called $X^{(2)}_G$.
    \begin{definition}[2-Strategy Space]
        $(\Xt,\dXt)$ is the set of all strategy profiles on the graph $G$ with only two strategies, together with the metric $\dXt(u,u')$ which tells the minimum number of single strategic changes to get from $u$ to $u'$.
    \end{definition}
    Likewise, we define the reduced 2-strategy space through equivalence under $\Phi$. 
    \begin{definition}[Reduced 2-Strategy Space]
        $(\At,\dAt)$ is the set of all equivalence classes under $\Phi$ of the set $\Xt$ together with the metric $\dAt$ which is defined by the adjacency relationship $\dAt(a,b)=1\iff \exists (u,u')\in (a,b)$ such that $\dXt(u,u')=1$.
    \end{definition}

    To define the two-strategy subgraph space, we need not define a new set because we will show that the well-known cyclespace is sufficient in this case. Recall the cyclespace of a graph $G$ is the set of all subgraphs of $G$ such that the subgraph has only even degree vertices (i.e., it admits an Eulerian circuit). We will equip this space with a metric which is defined similarly as the SFR.

    \begin{definition}[Cycle Single Face Rewiring]
        A Cycle Single Face Rewiring is a resampling of the graph $G$ to go from a subgraph $s$ to a subgraph $r$ which, for a single face $f$, satisfies the following 
        \begin{enumerate}
            \item \textbf{Locality} Any edge not incident to the face $f$ which was in $s$ is in $r$ and any edge not incident to $f$ which was not in $s$ is not in $r$
            \item \textbf{Closure} Any edge incident to the face $f$ which is in $s$ is not in $r$ and any edge incident to the face $f$ which is not in $s$ is in $r$
        \end{enumerate}
    \end{definition}

    We can take advantage of the fact that $C(G)$ is a vector space over the field $\mathbb{Z}_2$ to give a more compact definition. The definition of a CSFR states that every edge incident to $f$ becomes a non-edge and every non-edge incident to $f$ becomes an edge. This is equivalent to adding the cycle of edges incident to the face $f$ in $G$ to the existing subgraph. Therefore, we give the following corollary, which is direct from the definition 

    \begin{corollary}\label{cor:CSFRequiv}
        If $s$ and $r$ are subgraphs of $G$ and one can get from $s$ to $r$ by a CSFR of the face $f$ then $r=s+C_f$ in $C(G)$ where $C_f$ is the cycle of edges incident to the face $f$ in $G$.
    \end{corollary}

    \begin{definition}[2-Subgraph Space]
        $(C(G),d_C)$ is the set of all subgraphs of the graph $G$ which admit an Eulerian circuit together with the metric defined by CSFRs  
    \end{definition}

    \subsection{Nash equilibria in the two-strategy game}
    Because the metric space is slightly different (notice that not every CSFR is an SFR because there is no connectivity condition on the CSFR), we will show again that $\Phi\circ \Psi$ is an isometry.  For this main result, we must first ensure the functions are still well defined and bijective even between these different spaces.

    \begin{lemma}\label{lem:uniqueCutset}
        If $\mathcal{Q}_G^{(2)}$ is the set of partitions of the graph $G$ into two parts and $T(G)$ is the cutspace of the graph $G$ then $Cs:\mathcal{Q}_G^{(2)}\to T(G)$ which is defined by $Cs(P)$ is the cutset of $P$ is a bijection.
    \end{lemma}
    This result is well known, but for completeness we include the proof here. 
    \begin{proof}
        First observe that by the definition of a cut, the cutset of a partition into two parts is always in the cutspace, so this is a well-defined function. Moreover, this same definition says that every cut in the cutspace separates the graph into two parts. This means that for every cut  $t\in T(G)$ there is a partition $P\in \mathcal{Q}_G^{(2)}$ such that $Cs(P)=t$. Thus we need only show $Cs$ is injective. 
            
        To show injectivity, assume two partitions into two parts, $P$ and $Q$, are different. Without loss of generality, say there exists a pair of vertices $v,w$ which are in the same part in $P$ and in different parts in $Q$. If this does not happen, $P$ and $Q$ would be identical. Consider the shortest path from $v$ to $w$. The sequence of vertices can be labeled from $v$ to $w$ in order by which part they belong to, and when two adjacent vertices are not in the same part, they must have a cut edge between them. Let $Cs(P)$ be the cutset of the partition $P$ and let $Cs(Q)$ be the cutset of the partition $Q$. The path from $v$ to $w$ must have an even number of edges included in $C_P$. This is because either every vertex in the path is in the same part as $v$ and $w$ or, whenever the path enters the opposite part, it must exit before it reaches $w$. The same path must have an odd number of edges included in $C_Q$ because the path must leave the part that $v$ is in at least once, and if it enters it again, it must depart before it arrives at $w$.  This means that if $P\neq Q$ then $Cs(P)\neq Cs(Q)$.
    \end{proof}

    In addition to this well-known bijection, there is another well-known bijection between the cutspace $T(G)$ and the cyclespace of the dual, $C(G^*)$\cite{NishizekiPlanarGraphs}. We call this function $\Lambda$, but we notice that it is defined nearly exactly as we had defined $\Psi$. For a cut $t\in T(G)$, $\Lambda t$ is the subgraph spanned by the edges of $G^*$ corresponding to edges of $G$ included in $t$. This means that we can easily prove that $\Psi$ is a bijection.
    
    \begin{lemma}\label{lem:psiBijective}
        $\Psi:\mathcal{Q}_{G}^{(2)}\to C(G^*)$ is a bijection and can be expressed as $\Psi (P)=\Lambda \circ Cs(P)$. 
    \end{lemma}

    \begin{proof}
        Recall that $\Psi$ is defined as the function that takes a partition of $G$ and returns the subgraph of $G^*$ spanned by the edges of $G^*$ which correspond to edges of $G$ in the cutset of the original partition. $Cs(P)$ is the cut set of the original partition and $\Lambda$ takes a cut set from $T(G)$ and returns the subgraph of $G^*$ which is spanned by edges of the cut. 

        For any partition $P$, $\Psi P$, the subgraph spanned by the edges of the dual corresponding to the edges in $P$ can be described as the bijection $\Lambda \circ Cs(P)$ (see Lemma \ref {lem:uniqueCutset}) and as $\Psi P$. Thus we know $\Psi$ is bijective from $\mathcal{Q}_G^{(2)}$ to $C(G^*)$.
    \end{proof}

    These results allow us to prove the main ingredient in showing the stronger result in this restricted setting. Again, we will now show that $\Psi\circ\Phi:\At\to C(G^*)$ is an isometry.

    \begin{lemma}\label{lem:twoStratIsometry}
         $\Psi\circ\Phi:\At\to C(G^*)$ is an isometry. 
    \end{lemma}

    \begin{proof}
        Again we will argue that if $u$ and $u'$ are adjacent in $(\At,\dAt)$, then $c=\Psi\circ\Phi u$ and $c'=\Psi \circ \Phi u'$ are adjacent in $(C(G^*),d_C)$. To show $c$ and $c'$ are adjacent in this way, we must show that $c=c'+C_f$ in $\mathbb{Z}_2$ for some $f$ by Corollary \ref{cor:CSFRequiv}.
        
        If $u$ and $u'$ are adjacent in $(\At,\dAt)$, then there is a single vertex $v$ which changed their strategy. Without loss of generality, we can say they changed from strategy 1 to strategy 2. This means that the cutsets of the partitions $\Phi u$ and $\Phi u'$ differ only on edges adjacent to the vertex $v$.
        This is equivalent to saying that $Cs(\Phi u)-Cs(\Phi u')$ is a cut of $G$ which only has edges adjacent to $v$. Taking the isomorphism from the cutspace to the cyclespace of $G^*$ we can say 
        \[\Psi \circ \Phi u-\Psi \circ \Phi u'=\Lambda (Cs(\Phi u)-\Lambda (Cs(\Phi u'))=\Lambda (Cs(\Phi u)-Cs(\Phi u')) \]
        
        which implies that $c-c'$ is an element of the cyclespace of $G^*$ which only includes edges incident to the face $f$ of $G^*$ corresponding to the vertex $v$. Now observe that there is at least one edge in $c-c'$. If it did not, then $c-c'=0$ and, because of the bijectivity of $\Phi$ and $\Psi$ (lemma \ref{lem:psiBijective}), we would know $u=u'$ which is not the case. Because $c-c'$ is an element of the cyclespace of $G^*$, and because it is not empty, and because it only contains edges incident to a single face of $G^*$, it must contain every edge incident to that face. If it was missing an en edge, there would be a vertex with degree 1. Thus we know that $c-c'$ is the cycle of edges around the face $f$ which we call $C_F$. This means $c=c'+C_F$. Thus we know that if $\dAt(u,u')= 1$ then $d_C(\Psi\circ\Phi u ,\Psi\circ\Phi u')= 1$. 
        Because $\Psi \circ \Phi$ is bijective (lemma \ref{lem:psiBijective}) it must be also that $d_C(c,c')=1$ implies that $\dAt(\Phi^{-1}\circ\Psi^{-1} c,\Phi^{-1}\circ\Psi^{-1} c')=1$ and so our result it proved.
    \end{proof}

    This isometry is not a quasiisometry because there is no issue with non-uniqueness or with asymmetry. This is a full isometric bijective relation between the reduced $2-$strategy space and the cutspace of the dual. With this isometry, we can prove the stronger version of Theorem \ref{thm:minimalSubgraphs}

    \begin{theorem}\label{thm:2stratMinimalSubgraphs}
        For a planar graph $G$ with determined planar embedding and for its dual graph $G^*$, $s\in C(G^*)$ satisfies 
        \[size(s)\leq size(r)\, \forall r\in \Gamma_C(s), \]
        where $\Gamma_C(s):=\{r\in C(G^*);d_C(s,r)\leq 1\} $, if and only if $\Phi^{-1}\circ\Psi^{-1}s$ is a Nash equilibrium. 
    \end{theorem}

    \begin{proof}
         $\implies$ If $s$ is a local size minimizer in $C(G^*)$ then there are no faces $f$ for which $size(s+C_f)\leq size(f)$. By Lemma \ref{lem:twoStratIsometry}, this means that $u:=\Phi^{-1}\circ\Psi^{-1}s$ has fewer edges connecting vertices playing different strategies than any strategy profile in $\Gamma_{\At}(u).$

         Now, because of Lemma \ref{lem:Potentialissize}, we know also that $u$ has a higher potential than any strategy profile in its vicinity in $At$ and thus $u$ is a Nash equilibrium. 

        $\impliedby$ If $u$ is a Nash equilibrium, then by the definition of potential game $u$ has a higher potential than any strategy profile in its vicinity in $At$. By Lemmas \ref{lem:Potentialissize} and \ref{lem:twoStratIsometry}, $\Psi \circ \Phi u$ has fewer edges than any cycle in its vicinity in $C(G^*)$. This is equivalent to saying it is a local size minimizer and thus
        \[size(s)\leq size(r)\, \forall r\in \Gamma_C(s)\]
         Where $\Gamma_C(s):=\{r\in C(G^*);d_C(r,s)\leq 1\}$. This completes the proof.  
    \end{proof}

    This is a stronger result because not only do we recover what we already knew about minimal subgraphs (or cycles) indicating the presence of Nash equilibria. We also understand that Nash equilibria can be used to find minimal cycles in the cyclespace of the dual.

    The same isometry result which was used to prove theorem \ref{thm:2stratMinimalSubgraphs} can be used to prove an even stronger result in the dynamic case
    
    \subsection{Myopic Best Response in the two strategy case}
    Like many researchers considering the coordination game, we hope to be able to understand the dynamics when the game is repeated, and players are changing their strategies iteratively. The most standard way to understand this is through the strategy revision protocol called Myopic Best Response (MBR). This update rule simply says that every player observes the strategy profile at the present time step and, assuming that non of their co-players will change strategy, they update their own strategy to whichever strategy would maximize their individual payoff. Any sequence of strategy profiles which satisfies this rule, if it converges, must converge to a Nash equilibrium because every player is playing its best response.

    In general, MBR is not deterministic because there may be multiple strategies which maximize a player's payoff, and so some other information must be supplied to make a general MBR sequence into a deterministic sequence. In the case of two strategies, it is sufficient to say that, in the case of a tie, a player will not change its strategy. With this additional tie-breaking rule (which can be thought of as some $\epsilon$ cost to changing strategy), we can define an MBR operator which is well defined

    \begin{definition}[MBR operator]
        $f:\At \to \At$ is the operator which takes a strategy profile and returns the strategy profile resulting from a single iteration of Myopic Best Response. 
    \end{definition}

    These myopic best response sequences are a fundamental way to understand the dynamics of iterated games like this and have been employed to study structured coordination many times since the 1990s (e.g., \cite{Kandori1993,Ellison1993,Ellison2000,Oechssler1997,Oechssler1999} etc.)

    Again, although this is simple to understand, it is not always easy to imagine. The isometry between the strategy space and the cycle space in the dual allows us to discuss this behavior in a more easy to imagine, geometric way. This easier to imagine flow is called the Cycle Shortening Flow. We will define it, then demonstrate their equivalence through the isometry.

    \begin{definition}[Cycle Shortening Flow] For a cycle $c$ in the cyclespace of a planar graph $G$, let $g:C(G)\to C(G)$ be the operator defined as 
    \[g(c)=c+\sum_{f\in F(G^*)}H(c,f)C_F\]
    where $H(c,f)=1$ is more than half the edges of the face $f$ are included in $c$ and $0$ otherwise.

    A sequence of cycles $\{c_1,c_2,...\}$ satisfies the cycle shortening flow if and only if $c_t=g(c_{t-1})$.
    \end{definition}

    This flow is not obvious from its definition, but an illustrative example is included in Figure \ref{fig:SubgraphShortenningConsensus}. A more complete example can be seen in appendix \ref{app:CSFexample}. In short, the cycle changes to take the shortest path around every face that it touches. With this heuristic, it is easy to see that if these sequences converge, they must converge to Nash equilibria. 
    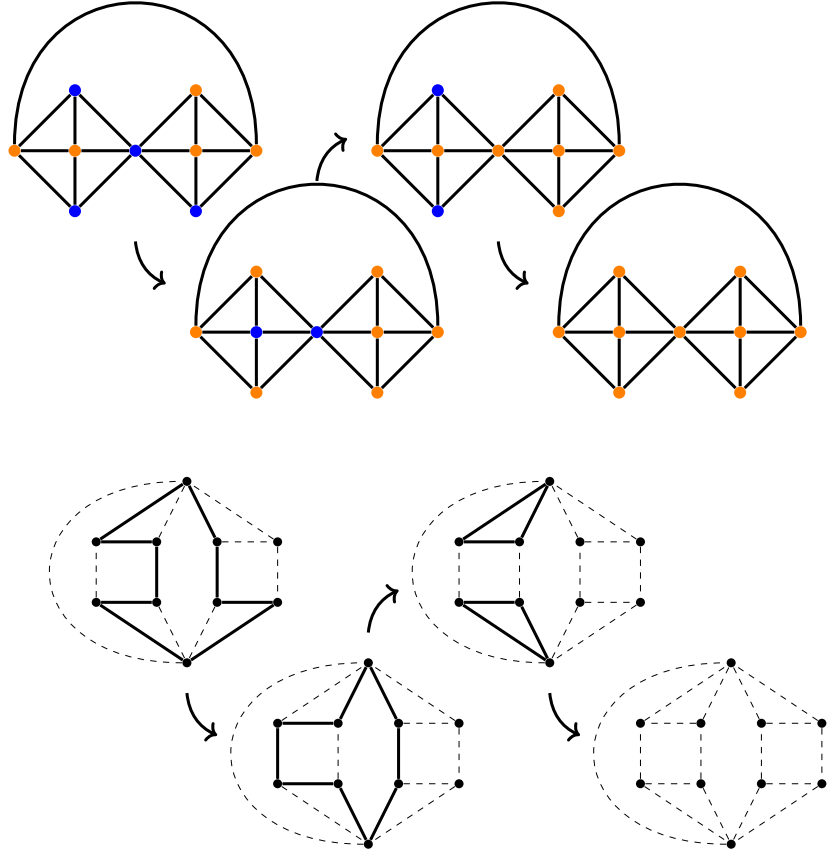
\begin{figure}
    \centering
    \scalebox{0.8}{
    \begin{tikzpicture}
        \node(a)[circle, fill=orange, inner sep =2pt] at (0,0){};
            \node(b)[circle, fill=orange, inner sep =2pt] at (1,0){};
            \node(c)[circle, fill=blue, inner sep =2pt] at (1,1){};
            \node(d)[circle, fill=blue, inner sep =2pt] at (1,-1){};
            \node(e)[circle, fill=blue, inner sep =2pt] at (2,0){};
            \node(f)[circle, fill=orange, inner sep =2pt] at (3,0){};
            \node(g)[circle, fill=orange, inner sep =2pt] at (3,1){};
            \node(h)[circle, fill=blue, inner sep =2pt] at (3,-1){};
            \node(i)[circle, fill=orange, inner sep =2pt] at (4,0){};

            \draw[line width = 0.5mm](a)--(b);
            \draw[line width = 0.5mm](a)--(c);
            \draw[line width = 0.5mm](a)--(d);
            \draw[line width = 0.5mm](b)--(c);
            \draw[line width = 0.5mm](b)--(d);
            \draw[line width = 0.5mm](b)--(e);
            \draw[line width = 0.5mm](c)--(e);
            \draw[line width = 0.5mm](d)--(e);
            \draw[line width = 0.5mm](e)--(f);
            \draw[line width = 0.5mm](e)--(g);
            \draw[line width = 0.5mm](e)--(h);
            \draw[line width = 0.5mm](f)--(g);
            \draw[line width = 0.5mm](f)--(h);
            \draw[line width = 0.5mm](f)--(i);
            \draw[line width = 0.5mm](g)--(i);
            \draw[line width = 0.5mm](h)--(i);
            \draw[line width = 0.5mm](a)to [out = 90,in = 90, looseness =2](i);

            \draw[->, line width = 0.5 mm] (2,-1.5)to [bend right = 30] (2.5, -2.2);

            \node(a)[circle, fill=orange, inner sep =2pt] at (3+0,-3+0){};
            \node(b)[circle, fill=blue, inner sep =2pt] at (3+1,-3+0){};
            \node(c)[circle, fill=orange, inner sep =2pt] at (3+1,-3+1){};
            \node(d)[circle, fill=orange, inner sep =2pt] at (3+1,-3+-1){};
            \node(e)[circle, fill=blue, inner sep =2pt] at (3+2,-3+0){};
            \node(f)[circle, fill=orange, inner sep =2pt] at (3+3,-3+0){};
            \node(g)[circle, fill=orange, inner sep =2pt] at (3+3,-3+1){};
            \node(h)[circle, fill=orange, inner sep =2pt] at (3+3,-3+-1){};
            \node(i)[circle, fill=orange, inner sep =2pt] at (3+4,-3+0){};

            \draw[line width = 0.5mm](a)--(b);
            \draw[line width = 0.5mm](a)--(c);
            \draw[line width = 0.5mm](a)--(d);
            \draw[line width = 0.5mm](b)--(c);
            \draw[line width = 0.5mm](b)--(d);
            \draw[line width = 0.5mm](b)--(e);
            \draw[line width = 0.5mm](c)--(e);
            \draw[line width = 0.5mm](d)--(e);
            \draw[line width = 0.5mm](e)--(f);
            \draw[line width = 0.5mm](e)--(g);
            \draw[line width = 0.5mm](e)--(h);
            \draw[line width = 0.5mm](f)--(g);
            \draw[line width = 0.5mm](f)--(h);
            \draw[line width = 0.5mm](f)--(i);
            \draw[line width = 0.5mm](g)--(i);
            \draw[line width = 0.5mm](h)--(i);
            \draw[line width = 0.5mm](a)to [out = 90,in = 90, looseness =2](i);

            \draw[->, line width = 0.5 mm] (3+2,-0.5)to [bend left= 30] (3+2.5, 0.2);

            \node(a)[circle, fill=orange, inner sep =2pt] at (6+0,-0+0){};
            \node(b)[circle, fill=orange, inner sep =2pt] at (6+1,-0+0){};
            \node(c)[circle, fill=blue, inner sep =2pt] at (6+1,-0+1){};
            \node(d)[circle, fill=blue, inner sep =2pt] at (6+1,-0+-1){};
            \node(e)[circle, fill=orange, inner sep =2pt] at (6+2,-0+0){};
            \node(f)[circle, fill=orange, inner sep =2pt] at (6+3,-0+0){};
            \node(g)[circle, fill=orange, inner sep =2pt] at (6+3,-0+1){};
            \node(h)[circle, fill=orange, inner sep =2pt] at (6+3,-0+-1){};
            \node(i)[circle, fill=orange, inner sep =2pt] at (6+4,-0+0){};

            \draw[line width = 0.5mm](a)--(b);
            \draw[line width = 0.5mm](a)--(c);
            \draw[line width = 0.5mm](a)--(d);
            \draw[line width = 0.5mm](b)--(c);
            \draw[line width = 0.5mm](b)--(d);
            \draw[line width = 0.5mm](b)--(e);
            \draw[line width = 0.5mm](c)--(e);
            \draw[line width = 0.5mm](d)--(e);
            \draw[line width = 0.5mm](e)--(f);
            \draw[line width = 0.5mm](e)--(g);
            \draw[line width = 0.5mm](e)--(h);
            \draw[line width = 0.5mm](f)--(g);
            \draw[line width = 0.5mm](f)--(h);
            \draw[line width = 0.5mm](f)--(i);
            \draw[line width = 0.5mm](g)--(i);
            \draw[line width = 0.5mm](h)--(i);
            \draw[line width = 0.5mm](a)to [out = 90,in = 90, looseness =2](i);
            
            \draw[->, line width = 0.5 mm] (6+2,-1.5)to [bend right= 30] (6+2.5, -2.2);

            \node(a)[circle, fill=orange, inner sep =2pt] at (9+0,-3+0){};
            \node(b)[circle, fill=orange, inner sep =2pt] at (9+1,-3+0){};
            \node(c)[circle, fill=orange, inner sep =2pt] at (9+1,-3+1){};
            \node(d)[circle, fill=orange, inner sep =2pt] at (9+1,-3+-1){};
            \node(e)[circle, fill=orange, inner sep =2pt] at (9+2,-3+0){};
            \node(f)[circle, fill=orange, inner sep =2pt] at (9+3,-3+0){};
            \node(g)[circle, fill=orange, inner sep =2pt] at (9+3,-3+1){};
            \node(h)[circle, fill=orange, inner sep =2pt] at (9+3,-3+-1){};
            \node(i)[circle, fill=orange, inner sep =2pt] at (9+4,-3+0){};

            \draw[line width = 0.5mm](a)--(b);
            \draw[line width = 0.5mm](a)--(c);
            \draw[line width = 0.5mm](a)--(d);
            \draw[line width = 0.5mm](b)--(c);
            \draw[line width = 0.5mm](b)--(d);
            \draw[line width = 0.5mm](b)--(e);
            \draw[line width = 0.5mm](c)--(e);
            \draw[line width = 0.5mm](d)--(e);
            \draw[line width = 0.5mm](e)--(f);
            \draw[line width = 0.5mm](e)--(g);
            \draw[line width = 0.5mm](e)--(h);
            \draw[line width = 0.5mm](f)--(g);
            \draw[line width = 0.5mm](f)--(h);
            \draw[line width = 0.5mm](f)--(i);
            \draw[line width = 0.5mm](g)--(i);
            \draw[line width = 0.5mm](h)--(i);
            \draw[line width = 0.5mm](a)to [out = 90,in = 90, looseness =2](i);

    \end{tikzpicture}
    }
    \vspace{1cm}
    
    \scalebox{0.8}{
    \begin{tikzpicture}
         \node(a)[circle, fill, inner sep =1.5pt] at (0,0){};
            \node(b)[circle, fill, inner sep =1.5pt] at (1,0){};
            \node(c)[circle, fill, inner sep =1.5pt] at (1,1){};
            \node(d)[circle, fill, inner sep =1.5pt] at (0,1){};
            \node(e)[circle, fill, inner sep =1.5pt] at (2,0){};
            \node(f)[circle, fill, inner sep =1.5pt] at (2,1){};
            \node(g)[circle, fill, inner sep =1.5pt] at (3,1){};
            \node(h)[circle, fill, inner sep =1.5pt] at (3,0){};
            \node(i)[circle, fill, inner sep =1.5pt] at (1.5,2){};
            \node(j)[circle, fill, inner sep =1.5pt] at (1.5,-1){};

            \draw[dashed](a)--(b);
            \draw[dashed](a)--(d);
            \draw[dashed](a)--(j);
            \draw[dashed](b)--(c);
            \draw[dashed](b)--(j);
            \draw[dashed](c)--(d);
            \draw[dashed](c)--(i);
            \draw[dashed](d)--(i);
            \draw[dashed](e)--(f);
            \draw[dashed](e)--(h);
            \draw[dashed](e)--(j);
            \draw[dashed](f)--(g);
            \draw[dashed](f)--(i);
            \draw[dashed](g)--(h);
            \draw[dashed](g)--(i);
            \draw[dashed](h)--(j);
            \draw[dashed](j)to [out = 180,in = 180, looseness = 2.5](i);

            \draw[line width = 0.5mm](a)--(b);
            \draw[line width = 0.5mm](b)--(c);
            \draw[line width = 0.5mm](c)--(d);
            \draw[line width = 0.5mm](d)--(i);
            \draw[line width = 0.5mm](i)--(f);
            \draw[line width = 0.5mm](f)--(e);
            \draw[line width = 0.5mm](e)--(h);
            \draw[line width = 0.5mm](h)--(j);
            \draw[line width = 0.5mm](j)--(a);

            \draw[->, line width = 0.5 mm] (1.5,-1.5)to [bend right = 30] (2, -2.2);

             \node(a)[circle, fill, inner sep =1.5pt] at (3+0,-3+0){};
            \node(b)[circle, fill, inner sep =1.5pt] at (3+1,-3+0){};
            \node(c)[circle, fill, inner sep =1.5pt] at (3+1,-3+1){};
            \node(d)[circle, fill, inner sep =1.5pt] at (3+0,-3+1){};
            \node(e)[circle, fill, inner sep =1.5pt] at (3+2,-3+0){};
            \node(f)[circle, fill, inner sep =1.5pt] at (3+2,-3+1){};
            \node(g)[circle, fill, inner sep =1.5pt] at (3+3,-3+1){};
            \node(h)[circle, fill, inner sep =1.5pt] at (3+3,-3+0){};
            \node(i)[circle, fill, inner sep =1.5pt] at (3+1.5,-3+2){};
            \node(j)[circle, fill, inner sep =1.5pt] at (3+1.5,-3+-1){};

            \draw[dashed](a)--(b);
            \draw[dashed](a)--(d);
            \draw[dashed](a)--(j);
            \draw[dashed](b)--(c);
            \draw[dashed](b)--(j);
            \draw[dashed](c)--(d);
            \draw[dashed](c)--(i);
            \draw[dashed](d)--(i);
            \draw[dashed](e)--(f);
            \draw[dashed](e)--(h);
            \draw[dashed](e)--(j);
            \draw[dashed](f)--(g);
            \draw[dashed](f)--(i);
            \draw[dashed](g)--(h);
            \draw[dashed](g)--(i);
            \draw[dashed](h)--(j);
            \draw[dashed](j)to [out = 180,in = 180, looseness = 2.5](i);

            \draw[line width = 0.5mm](a)--(b);
            \draw[line width = 0.5mm](b)--(j);
            \draw[line width = 0.5mm](j)--(e);
            \draw[line width = 0.5mm](e)--(f);
            \draw[line width = 0.5mm](f)--(i);
            \draw[line width = 0.5mm](i)--(c);
            \draw[line width = 0.5mm](c)--(d);
            \draw[line width = 0.5mm](d)--(a);

            \draw[->, line width = 0.5 mm] (3+1.5,-0.5)to [bend left= 30] (3+2, 0.2);

                         \node(a)[circle, fill, inner sep =1.5pt] at (6+0,-0+0){};
            \node(b)[circle, fill, inner sep =1.5pt] at (6+1,-0+0){};
            \node(c)[circle, fill, inner sep =1.5pt] at (6+1,-0+1){};
            \node(d)[circle, fill, inner sep =1.5pt] at (6+0,-0+1){};
            \node(e)[circle, fill, inner sep =1.5pt] at (6+2,-0+0){};
            \node(f)[circle, fill, inner sep =1.5pt] at (6+2,-0+1){};
            \node(g)[circle, fill, inner sep =1.5pt] at (6+3,-0+1){};
            \node(h)[circle, fill, inner sep =1.5pt] at (6+3,-0+0){};
            \node(i)[circle, fill, inner sep =1.5pt] at (6+1.5,-0+2){};
            \node(j)[circle, fill, inner sep =1.5pt] at (6+1.5,-0+-1){};

            \draw[dashed](a)--(b);
            \draw[dashed](a)--(d);
            \draw[dashed](a)--(j);
            \draw[dashed](b)--(c);
            \draw[dashed](b)--(j);
            \draw[dashed](c)--(d);
            \draw[dashed](c)--(i);
            \draw[dashed](d)--(i);
            \draw[dashed](e)--(f);
            \draw[dashed](e)--(h);
            \draw[dashed](e)--(j);
            \draw[dashed](f)--(g);
            \draw[dashed](f)--(i);
            \draw[dashed](g)--(h);
            \draw[dashed](g)--(i);
            \draw[dashed](h)--(j);
            \draw[dashed](j)to [out = 180,in = 180, looseness = 2.5](i);

            \draw[line width = 0.5mm](a)--(b);
            \draw[line width = 0.5mm](b)--(j);
            \draw[line width = 0.5mm](j)--(a);
            \draw[line width = 0.5mm](d)--(i);
            \draw[line width = 0.5mm](i)--(c);
            \draw[line width = 0.5mm](c)--(d);

            \draw[->, line width = 0.5 mm] (6+1.5,-1.5)to [bend right= 30] (6+2, -2.2);

            \node(a)[circle, fill, inner sep =1.5pt] at (9+0,-3+0){};
            \node(b)[circle, fill, inner sep =1.5pt] at (9+1,-3+0){};
            \node(c)[circle, fill, inner sep =1.5pt] at (9+1,-3+1){};
            \node(d)[circle, fill, inner sep =1.5pt] at (9+0,-3+1){};
            \node(e)[circle, fill, inner sep =1.5pt] at (9+2,-3+0){};
            \node(f)[circle, fill, inner sep =1.5pt] at (9+2,-3+1){};
            \node(g)[circle, fill, inner sep =1.5pt] at (9+3,-3+1){};
            \node(h)[circle, fill, inner sep =1.5pt] at (9+3,-3+0){};
            \node(i)[circle, fill, inner sep =1.5pt] at (9+1.5,-3+2){};
            \node(j)[circle, fill, inner sep =1.5pt] at (9+1.5,-3+-1){};

            \draw[dashed](a)--(b);
            \draw[dashed](a)--(d);
            \draw[dashed](a)--(j);
            \draw[dashed](b)--(c);
            \draw[dashed](b)--(j);
            \draw[dashed](c)--(d);
            \draw[dashed](c)--(i);
            \draw[dashed](d)--(i);
            \draw[dashed](e)--(f);
            \draw[dashed](e)--(h);
            \draw[dashed](e)--(j);
            \draw[dashed](f)--(g);
            \draw[dashed](f)--(i);
            \draw[dashed](g)--(h);
            \draw[dashed](g)--(i);
            \draw[dashed](h)--(j);
            \draw[dashed](j)to [out = 180,in = 180, looseness = 2.5](i);

    \end{tikzpicture}}
    \caption[Subgraph shortening to consensus]{\textbf{Top :} A trajectory in the dynamic pure coordination game on the graph $G$ with two strategies that leads to a consensus equilibrium. \textbf{Bottom :} The corresponding cycle shortening flow in the dual $G^*$. The dotted edges are edges of $G^*$ which are not included in the subgraph $s_t$, and the bold edges are included in $s_t$.}
    \label{fig:SubgraphShortenningConsensus}
\end{figure}
    
Having defined these two flows, one entirely in the strategy space and one entirely in the cycle space, we now can show that if a sequence satisfies one flow, its image (or preimage) under the isometry satisfies the other. 

\begin{theorem}\label{thm:csfEquiv}
    The dynamics in $\At$ of the pure coordination game under Myopic best response are equivalent to the Cycle Shortening Flow on $C(G^*)$. More specifically:  Given a $u_0\in \At$ then the flows $u(t+1)=f(u(t)), u(0)=u_0$ and $s_{t+1}=g(s_t), g(0)=\Psi\circ\Phi u_0$ satisfy $\Psi\circ\Phi u(t)=s_t$ for all $t$.
\end{theorem}

\begin{proof}
    The only thing necessary to prove is that $f(u)=\Phi^{-1}\circ\Psi^{-1}(g(\Psi\circ\Phi u))$ Starting from the strategy profile $u$, if $u'=f(u)$ that means that if $v$ was playing a best response in $u$ then it did not change strategies and so it is playing the same strategy in $u'$. If it was not playing a best response in $u$, then it is playing the opposite strategy in $u'$. Notice that if two players are playing the same strategy and they both change strategies, in the two-strategy case, they are still playing the same strategy. Likewise, if they are playing different strategies and both change strategies, then they are still playing different strategies. This is to say, if both players change strategies, then if the edge between those two players was in $Cs(\Phi u)$, then it is also in $Cs(\Phi u')$. Moreover, if it was not in $Cs(\Phi u)$, it will not be in $Cs(\Phi u')$. Likewise, if only one of two adjacent players changes its strategy, then if the edge was in $Cs(\Phi u)$, it will not be in $Cs(\Phi u')$, and if it was not in $Cs(\Phi u)$, it will be in $Cs(\Phi u')$.

    This means we have captured the difference in the edges contained in $\Psi\circ \Phi u$ and $\Psi \circ\Phi f(u)$. Consider any edge in $G^*$. It separates two faces $f_1$ and $f_2$, which correspond to two players $v_1$ and $v_2$ in $G$. If exactly one of these vertices is not playing a best response, then the edge changes its status (either it was included in $\Psi \circ \Phi u$ and is no longer included in $\Psi \circ\Phi f(u)$ or it was not included in $\Psi \circ \Phi u$ and is now included in $\Psi \circ \Phi f(u)$). If both or neither of the vertices are playing a best response, the edge does not change its status.

    Observe that if a player, $v$, is not playing a best response in $u$, then, in the dual, the corresponding face $f$ will have more edges than non-edges in $\Psi\circ\Phi u$. This is exactly the condition from the cycle shortening flow for when $C_f$ is added to the subgraph. Using this logic, and the definition of the Cycle Shortening Flow, we can say the following about $\Psi\circ \Phi u$ and $g(\Psi\circ\Phi u)$: Consider an edge in $G^*$ which separates two faces $f_1$ and $f_2$. These faces correspond to two adjacent vertices in $G$, $v_1$ and $v_2$. If exactly one of these vertices is not playing a best response then exactly one of $C_{f_1}$ or $C_{f_2}$ is added to $\Psi\circ\Phi u$ and so (by addition in the field $\mathbb{Z}_2$) if the edge was in $\Psi\circ\Phi u$ it is not it $g(\Psi \circ\Phi u)$ and if it was not in $\Psi\circ\Phi u$ then it is in $g(\Psi\circ\Phi u)$. Likewise, if both or neither of $v_1$ and $v_2$ are playing a best response in $u$, then if the edge is in $\Psi\circ\Phi u$, it will remain in $g(\Psi\circ\Phi u)$, and if it was not in $\Psi\circ\Phi u$, then it will not be in $g(\Psi\circ\Phi u)$.

    This is a long winded way of saying that, given any $u\in \At$, If an edge is in $\Psi\circ\Phi f(u)$, then it will also be in $g(\Psi \circ\Phi u)$, and if an edge is missing from $\Psi \circ\Phi f(u)$, then it will also be absent from $g(\Psi\circ\Phi u)$. This means that these subgraphs are identical and we can conclude, by Lemma \ref{lem:psiBijective} that,
    \[f(u)=\Phi^{-1}\circ\Psi^{-1}g(\Psi\circ\Phi u).\]
\end{proof}

This result shows that there is a dynamic graph process, which can be defined entirely without reference to the game theory, which recapitulates the behavior of MBR in the pure coordination game. Moreover, it gives a second way to prove theorem \ref{thm:2stratMinimalSubgraphs} because Nash equilibria are stationary solutions of the MBR flow which correspond to stationary solutions of the Cycle Shortening Flow which must be minimal cycles. These results, together with the more general, less powerful theorem \ref{thm:minimalSubgraphs}, give us some geometric interpretations of coordination on planar graphs.

\section{Discussion}\label{sec:discussion}

Thinking of the pure coordination game on planar graphs through the dual gives us a much better geometric intuition for the structures that admit non-consensus equilibria. By proving that there is a quasiisometry between the strategy space $\mathcal{A}_G$ and the space of bridge-free subgraphs in the dual, $\mathcal{SC}_{G^*}$,  we have demonstrated that a locally minimal subgraph always corresponds to a Nash equilibrium in the original graph. 

This correspondence helps us understand several things about the pure coordination game. First, it makes the connection between the pure coordination game and the mincut partition quite obvious. Although this connection can be made without the use of the dual, the dual approach allows us to see that equilibrium partitions are to mincut partitions as local optima are to global optima. Excepting the degenerate case in which a part of a mincut partition has a single vertex, this makes the set containment immediate. In demonstrating this relationship, it also enforces our understanding of coordination as a type of local community detection. 
The reason that this is important is that it gives us a locally driven concept of community detection. Instead of finding communities by maximizing a global objective, we can find communities by maximizing a local objective at every point. The use of the dual approach turns the local maximization problem into a global minimization problem. Therefore, through the use of the dual, we can describe a local community detection method through a global optimization problem. 

This type of locally driven community detection is becoming more and more important in the age of autonomous agents. When autonomous agents interact with one another in a coordination-like setting, finding communities in the network of agents cannot always be directed by a top-down, global optimization approach. Instead, if each agent is making decisions autonomously, understanding the community structure of such a network requires a local concept of community detection. These findings show that finding locally driven communities in a planar assemblage of autonomous agents is equivalent to finding minimal subgraphs in the dual. 

Although we can think of this as a global optimization problem in the dual sense, we get no added computational power through this new conception. Because the locality in the space of subgraphs is defined by local perturbations, finding locally minimal subgraphs in the dual is not more computationally efficient than finding Nash equilibria in the original game. The dual concept only gives us a better geometric intuition for the usefulness of these coordination game equilibria. 

The correspondence between Nash equilibria and minimal subgraphs also helps us imagine the kinds of planar graphs which admit non-consensus equilibria. If, in the dual, every cycle has at least one face, $f$, in which the majority of the edges of $f$ are included in the cycle, then the graph $G$ cannot admit a non-consensus equilibrium with two strategies. The same statement for $m>2$ strategies is more complicated, but heuristically we can think about trying to stretch a rubber band (or weird complex of rubber bands) around a polyhedron defined by a planar graph. If that polyhedron is shaped so that, no matter how we stretch the rubber bands along its edges, it will eventually slide off, then that planar graph admits only the consensus equilibrium.

This heuristic is only truly correct in the continuum limit of the planar graph, but it allows us to think about the structures required of the dual graph to admit non-consensus equilibrium.  Again, this way of conceptualizing Nash equilibria is easier for us to imagine than for computers to solve because it involves identifying minimal cycles with a locality which feels quite natural to us but is not yet clear algorithmically. 

Although this method does not improve our ability to understand the discrete game algorithmically, it does make the connection to the continuous space game quite clear. It is beyond the scope of this manuscript, but the cycle shortening flow has clear similarities to the curve shortening flow. In the same way that Stationary solutions of the curve shortening flow are minimal surfaces, stationary solutions to the cycle shortening flow are minimal subgraphs and thus Nash equilibria. Further study of how this game may translate to the continuous setting may show many similarities to the curve shortening flow and to the network flow.

Being able to imagine the strategy space with a certain metric and being able to describe the dynamic game as flows through the strategy space makes it possible to draw equivalences between the game dynamics and other better understood flows, or at least other flows which are easier to define. These connections help us get a better intuitive understanding of the game dynamics and, in some cases, can show us where non-consensus equilibria can emerge. 
\bibliography{references}
\appendix
\section{Cycle Shortening Flow Example}\label{app:CSFexample}
The following several figures are a sequence of subgraphs from the cycle space of the graph $G$ which are obeying cycle shortening flow from section \ref{sec:mbr}. At each time step, for each face $f$, $C_f$ (the cycle of edges incident to $f$) is either added or not added to the subgraph depending on the present number of edges adjacent to $f$. If there are more edges than non-edges of $f$ included in the subgraph then $C_f$ is added. Otherwise, it is not. From this, the cycle shortening flow is very simple. 

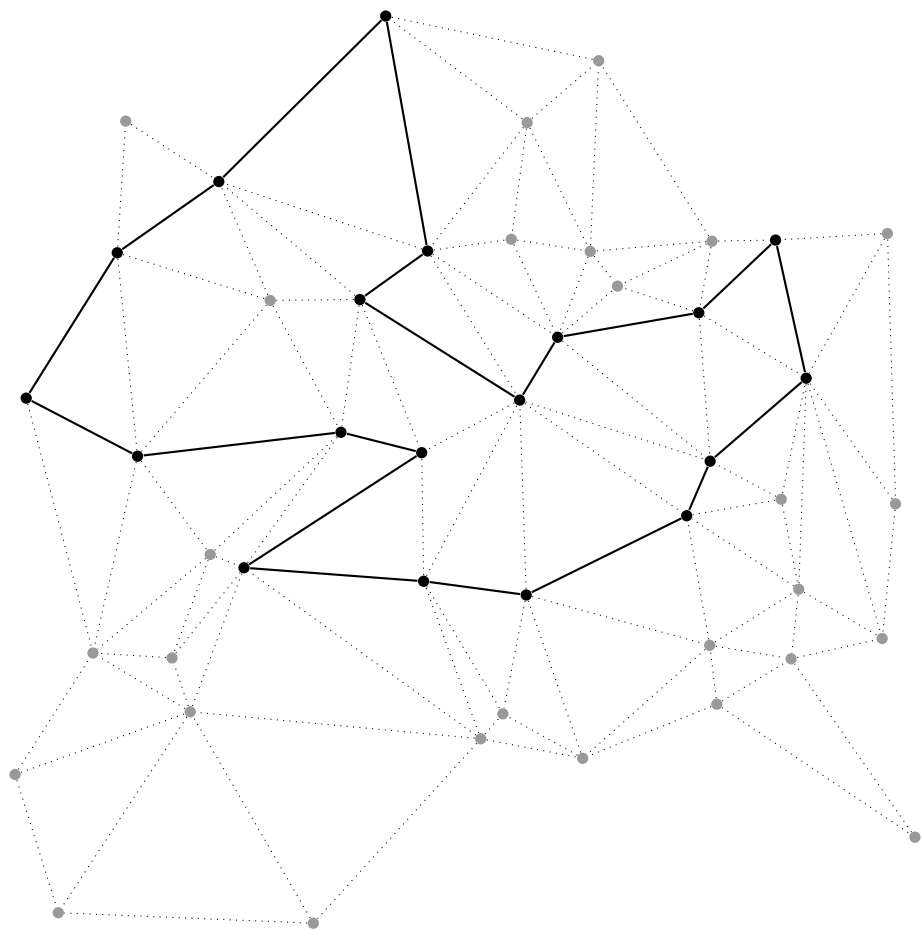
\begin{figure}
    \centering
  \begin{tikzpicture}[scale =1.5]
    \begin{scope}[circle, fill=black, draw=black, inner sep = 1.5pt]
      \draw[black]
        (6.036, 5.662) node[fill = black, opacity = 0.4](1){}
        (7.583, 5.73) node[fill = black, opacity = 0.4] (2){}
        (2.931, 5.148) node [fill = black](3){}
        (6.079, 1.579) node[fill = black, opacity = 0.4] (4){}
        (3.997, 1.274) node [fill = black, opacity = 0.4](5){}
        (-0.012, 4.278) node[fill = black] (6){}
        (4.406, 6.708) node[fill = black, opacity = 0.4] (7){}
        (4.398, 2.542) node[fill = black] (9){}
        (7.655, 3.346) node[fill = black, opacity = 0.4] (10){}
        (2.52, -0.353) node [fill = black, opacity = 0.4](11){}
        (0.791, 5.56) node[fill = black] (12){}
        (0.97, 3.766) node[fill = black] (13){}
        (7.826, 0.406) node [fill = black, opacity = 0.4](15){}
        (0.866, 6.721) node[fill = black, opacity = 0.4] (16){}
        (3.493, 2.663) node [fill = black](17){}
        (-0.11, 0.959) node[fill = black, opacity = 0.4] (18){}
        (6.016, 2.098) node[fill = black, opacity = 0.4] (19){}
        (3.528, 5.576) node[fill = black] (20){}
        (1.611, 2.899) node[fill = black, opacity = 0.4] (21){}
        (2.14, 5.139) node [fill = black, opacity = 0.4](22){}
        (3.159, 7.648) node[fill = black] (24){}
        (6.647, 3.386) node[fill = black, opacity = 0.4] (25){}
        (6.801, 2.593) node[fill = black, opacity = 0.4] (26){}
        (1.435, 1.512) node[fill = black, opacity = 0.4] (27){}
        (1.687, 6.188) node [fill = black](28){}
        (6.597, 5.672) node[fill = black] (29){}
        (4.896, 1.102) node[fill = black, opacity = 0.4] (30){}
        (1.274, 1.987) node[fill = black, opacity = 0.4] (31){}
        (5.92, 5.031) node[fill = black] (32){}
        (1.909, 2.782) node [fill = black](33){}
        (4.266, 5.679) node [fill = black, opacity = 0.4](34){}
        (4.962, 5.572) node[fill = black, opacity = 0.4] (35){}
        (5.037, 7.254) node[fill = black, opacity = 0.4] (36){}
        (4.34, 4.261) node[fill = black] (37){}
        (4.675, 4.815) node[fill = black] (38){}
        (6.02, 3.722) node [fill = black](39){}
        (0.271, -0.26) node[fill = black, opacity = 0.4] (40){}
        (5.203, 5.266) node[fill = black, opacity = 0.4] (41){}
        (5.813, 3.241) node[fill = black] (42){}
        (0.577, 2.03) node[fill = black, opacity = 0.4] (43){}
        (6.867, 4.454) node[fill = black] (44){}
        (6.734, 1.979) node [fill = black, opacity = 0.4](45){}
        (4.192, 1.494) node[fill = black, opacity = 0.4] (46){}
        (7.537, 2.158) node [fill = black, opacity = 0.4](47){}
        (3.476, 3.797) node [fill = black](48){}
        (2.764, 3.977) node [fill = black](49){};
        \end{scope}
      \begin{scope}[thick]
        \draw (6) to (13);
        \draw (13) to (49);
        \draw (48) to (49);
        \draw (33) to (48);

        \draw (3) to (37);
        \draw (37) to (38);
        
        \draw (39) to (42);
        \draw (9) to (42);
        \draw (9) to (17);
        \draw (17) to (33);

        \draw (3) to (20);
        \draw (20) to (24);
        \draw (24) to (28);
        \draw (12) to (28);
        \draw (6) to (12);

        \draw (32) to (38);
        
        \draw (29) to (32);
        
        \draw (29) to (44);
        
        \draw (39) to (44);

      \end{scope}
        
      \begin{scope}[dotted]
        \draw (1) to (29);
        \draw (1) to (32);
        \draw (1) to (35);
        \draw (1) to (36);
        \draw (1) to (41);
        \draw (2) to (10);
        \draw (2) to (29);
        \draw (2) to (44);
        
        \draw (3) to (22);
        \draw (3) to (28);
        
        \draw (3) to (48);
        \draw (3) to (49);
        \draw (4) to (15);
        \draw (4) to (19);
        \draw (4) to (30);
        \draw (4) to (45);
        \draw (5) to (11);
        \draw (5) to (17);
        \draw (5) to (27);
        \draw (5) to (30);
        \draw (5) to (33);
        \draw (5) to (46);

        \draw (6) to (43);
        \draw (7) to (20);
        \draw (7) to (24);
        \draw (7) to (34);
        \draw (7) to (35);
        \draw (7) to (36);
        
        \draw (9) to (19);
        
        \draw (9) to (30);
        \draw (9) to (37);
        \draw (9) to (46);
        \draw (10) to (44);
        \draw (10) to (47);
        \draw (11) to (27);
        \draw (11) to (40);
        \draw (12) to (13);
        \draw (12) to (16);
        \draw (12) to (22);
        
        \draw (13) to (21);
        \draw (13) to (22);
        \draw (13) to (43);
        
        \draw (15) to (45);
        \draw (16) to (28);
        
        \draw (17) to (37);
        \draw (17) to (46);
        \draw (17) to (48);
        \draw (18) to (27);
        \draw (18) to (40);
        \draw (18) to (43);
        \draw (19) to (26);
        \draw (19) to (30);
        \draw (19) to (42);
        \draw (19) to (45);
        
        \draw (20) to (28);
        \draw (20) to (34);
        \draw (20) to (37);
        \draw (20) to (38);
        \draw (21) to (31);
        \draw (21) to (33);
        \draw (21) to (43);
        \draw (21) to (49);
        \draw (22) to (28);
        \draw (22) to (49);
        
        \draw (24) to (36);
        \draw (25) to (26);
        \draw (25) to (39);
        \draw (25) to (42);
        \draw (25) to (44);
        \draw (26) to (42);
        \draw (26) to (44);
        \draw (26) to (45);
        \draw (26) to (47);
        \draw (27) to (31);
        \draw (27) to (33);
        \draw (27) to (40);
        \draw (27) to (43);

        \draw (30) to (46);
        \draw (31) to (33);
        \draw (31) to (43);
        
        \draw (32) to (39);
        \draw (32) to (41);
        \draw (32) to (44);
        
        \draw (33) to (49);
        \draw (34) to (35);
        \draw (34) to (38);
        \draw (35) to (36);
        \draw (35) to (38);
        \draw (35) to (41);
        
        \draw (37) to (39);
        \draw (37) to (42);
        \draw (37) to (48);
        \draw (38) to (39);
        \draw (38) to (41);

        \draw (44) to (47);
        \draw (45) to (47);
        
      \end{scope}
    \end{tikzpicture}

    \caption[Cycle shortening flow frame 1]{An initial subgraph (in the solid lines) from the cycle space of the graph shown in the dotted lines}
    \label{fig:CSF1}
\end{figure}

\begin{figure}
    \centering
  \begin{tikzpicture}[scale =1.5]
      \begin{scope}[circle, fill=black, draw=black, inner sep = 1.5pt]
      \draw[black]
        (6.036, 5.662) node[fill = black, opacity = 0.4](1){}
        (7.583, 5.73) node[fill = black, opacity = 0.4] (2){}
        (2.931, 5.148) node [fill = black, opacity = 0.4](3){}
        (6.079, 1.579) node[fill = black, opacity = 0.4] (4){}
        (3.997, 1.274) node [fill = black, opacity = 0.4](5){}
        (-0.012, 4.278) node[fill = black, opacity = 0.4] (6){}
        (4.406, 6.708) node[fill = black, opacity = 0.4] (7){}
        (4.398, 2.542) node[fill = black] (9){}
        (7.655, 3.346) node[fill = black, opacity = 0.4] (10){}
        (2.52, -0.353) node [fill = black, opacity = 0.4](11){}
        (0.791, 5.56) node[fill = black] (12){}
        (0.97, 3.766) node[fill = black] (13){}
        (7.826, 0.406) node [fill = black, opacity = 0.4](15){}
        (0.866, 6.721) node[fill = black, opacity = 0.4] (16){}
        (3.493, 2.663) node [fill = black](17){}
        (-0.11, 0.959) node[fill = black, opacity = 0.4] (18){}
        (6.016, 2.098) node[fill = black, opacity = 0.4] (19){}
        (3.528, 5.576) node[fill = black] (20){}
        (1.611, 2.899) node[fill = black, opacity = 0.4] (21){}
        (2.14, 5.139) node [fill = black, opacity = 0.4](22){}
        (3.159, 7.648) node[fill = black, opacity = 0.4] (24){}
        (6.647, 3.386) node[fill = black, opacity = 0.4] (25){}
        (6.801, 2.593) node[fill = black, opacity = 0.4] (26){}
        (1.435, 1.512) node[fill = black, opacity = 0.4] (27){}
        (1.687, 6.188) node [fill = black](28){}
        (6.597, 5.672) node[fill = black, opacity = 0.4] (29){}
        (4.896, 1.102) node[fill = black, opacity = 0.4] (30){}
        (1.274, 1.987) node[fill = black, opacity = 0.4] (31){}
        (5.92, 5.031) node[fill = black] (32){}
        (1.909, 2.782) node [fill = black](33){}
        (4.266, 5.679) node [fill = black, opacity = 0.4](34){}
        (4.962, 5.572) node[fill = black, opacity = 0.4] (35){}
        (5.037, 7.254) node[fill = black, opacity = 0.4] (36){}
        (4.34, 4.261) node[fill = black] (37){}
        (4.675, 4.815) node[fill = black] (38){}
        (6.02, 3.722) node [fill = black](39){}
        (0.271, -0.26) node[fill = black, opacity = 0.4] (40){}
        (5.203, 5.266) node[fill = black, opacity = 0.4] (41){}
        (5.813, 3.241) node[fill = black] (42){}
        (0.577, 2.03) node[fill = black, opacity = 0.4] (43){}
        (6.867, 4.454) node[fill = black] (44){}
        (6.734, 1.979) node [fill = black, opacity = 0.4](45){}
        (4.192, 1.494) node[fill = black, opacity = 0.4] (46){}
        (7.537, 2.158) node [fill = black, opacity = 0.4](47){}
        (3.476, 3.797) node [fill = black](48){}
        (2.764, 3.977) node [fill = black](49){};
        \end{scope}
       \begin{scope}[thick]
        \draw (9) to (17);
        \draw (9) to (42);
        \draw (12) to (13);
        \draw (12) to (28);
        \draw (13) to (49);
        \draw (17) to (48);
        \draw (20) to (28);
        \draw (20) to (37);
        \draw (32) to (38);
        \draw (32) to (44);
        \draw (33) to (48);
        \draw (33) to (49);
        \draw (37) to (38);
        \draw (39) to (44);
        \draw (39) to (42);
      \end{scope}
        
      \begin{scope}[dotted]
        \draw (1) to (29);
        \draw (1) to (32);
        \draw (1) to (35);
        \draw (1) to (36);
        \draw (1) to (41);
        \draw (2) to (10);
        \draw (2) to (29);
        \draw (2) to (44);
        \draw (3) to (20);
        \draw (3) to (22);
        \draw (3) to (28);
        \draw (3) to (37);
        \draw (3) to (48);
        \draw (3) to (49);
        \draw (4) to (15);
        \draw (4) to (19);
        \draw (4) to (30);
        \draw (4) to (45);
        \draw (5) to (11);
        \draw (5) to (17);
        \draw (5) to (27);
        \draw (5) to (30);
        \draw (5) to (33);
        \draw (5) to (46);
        \draw (6) to (12);
        \draw (6) to (13);
        \draw (6) to (43);
        \draw (7) to (20);
        \draw (7) to (24);
        \draw (7) to (34);
        \draw (7) to (35);
        \draw (7) to (36);
        \draw (9) to (19);
        \draw (9) to (30);
        \draw (9) to (37);
        \draw (9) to (46);
        \draw (10) to (44);
        \draw (10) to (47);
        \draw (11) to (27);
        \draw (11) to (40);
        \draw (12) to (16);
        \draw (12) to (22);
        \draw (13) to (21);
        \draw (13) to (22);
        \draw (13) to (43);
        \draw (15) to (45);
        \draw (16) to (28);
        \draw (17) to (33);
        \draw (17) to (37);
        \draw (17) to (46);
        \draw (18) to (27);
        \draw (18) to (40);
        \draw (18) to (43);
        \draw (19) to (26);
        \draw (19) to (30);
        \draw (19) to (42);
        \draw (19) to (45);
        \draw (20) to (24);
        \draw (20) to (34);
        \draw (20) to (38);
        \draw (21) to (31);
        \draw (21) to (33);
        \draw (21) to (43);
        \draw (21) to (49);
        \draw (22) to (28);
        \draw (22) to (49);
        \draw (24) to (28);
        \draw (24) to (36);
        \draw (25) to (26);
        \draw (25) to (39);
        \draw (25) to (42);
        \draw (25) to (44);
        \draw (26) to (42);
        \draw (26) to (44);
        \draw (26) to (45);
        \draw (26) to (47);
        \draw (27) to (31);
        \draw (27) to (33);
        \draw (27) to (40);
        \draw (27) to (43);
        \draw (29) to (32);
        \draw (29) to (44);
        \draw (30) to (46);
        \draw (31) to (33);
        \draw (31) to (43);
        \draw (32) to (39);
        \draw (32) to (41);
        \draw (34) to (35);
        \draw (34) to (38);
        \draw (35) to (36);
        \draw (35) to (38);
        \draw (35) to (41);
        \draw (37) to (39);
        \draw (37) to (42);
        \draw (37) to (48);
        \draw (38) to (39);
        \draw (38) to (41);
        \draw (44) to (47);
        \draw (45) to (47);
        \draw (48) to (49);
      \end{scope}
    \end{tikzpicture}

    \caption[Cycle shortening flow frame 2]{After the first step of the cycle shortening flow, for every face $f$ of $G$ with more edges than non-edges in the subgraph, the cycle $C_f$ is added to the subgraph. }
    \label{fig:CSF2}
\end{figure}
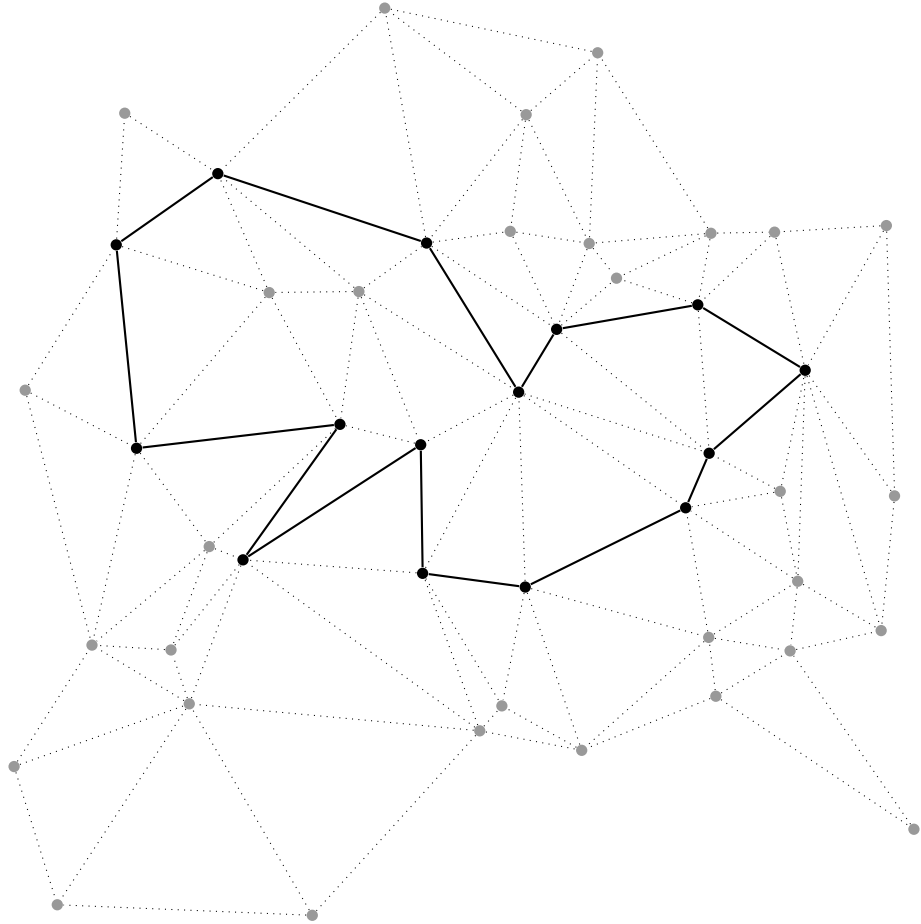

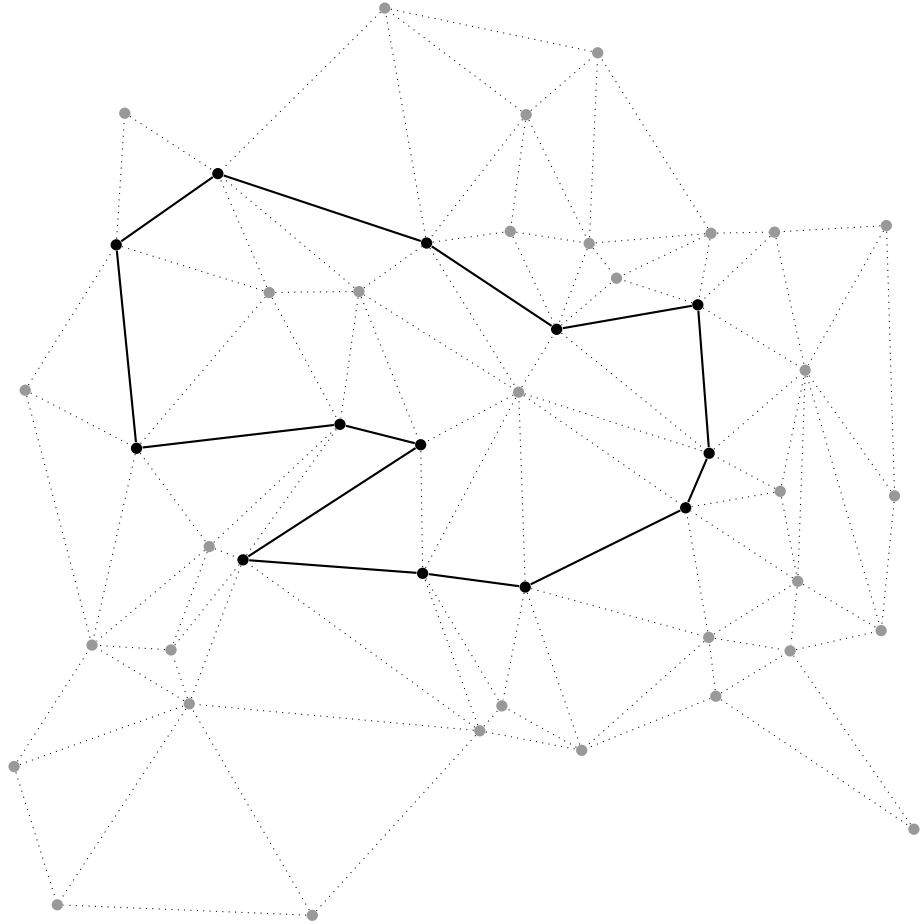
\begin{figure}
    \centering
  \begin{tikzpicture}[scale =1.5]
       \begin{scope}[circle, fill=black, draw=black, inner sep = 1.5pt]
      \draw[black]
        (6.036, 5.662) node[fill = black, opacity = 0.4](1){}
        (7.583, 5.73) node[fill = black, opacity = 0.4] (2){}
        (2.931, 5.148) node [fill = black, opacity = 0.4](3){}
        (6.079, 1.579) node[fill = black, opacity = 0.4] (4){}
        (3.997, 1.274) node [fill = black, opacity = 0.4](5){}
        (-0.012, 4.278) node[fill = black, opacity = 0.4] (6){}
        (4.406, 6.708) node[fill = black, opacity = 0.4] (7){}
        (4.398, 2.542) node[fill = black] (9){}
        (7.655, 3.346) node[fill = black, opacity = 0.4] (10){}
        (2.52, -0.353) node [fill = black, opacity = 0.4](11){}
        (0.791, 5.56) node[fill = black] (12){}
        (0.97, 3.766) node[fill = black] (13){}
        (7.826, 0.406) node [fill = black, opacity = 0.4](15){}
        (0.866, 6.721) node[fill = black, opacity = 0.4] (16){}
        (3.493, 2.663) node [fill = black](17){}
        (-0.11, 0.959) node[fill = black, opacity = 0.4] (18){}
        (6.016, 2.098) node[fill = black, opacity = 0.4] (19){}
        (3.528, 5.576) node[fill = black] (20){}
        (1.611, 2.899) node[fill = black, opacity = 0.4] (21){}
        (2.14, 5.139) node [fill = black, opacity = 0.4](22){}
        (3.159, 7.648) node[fill = black, opacity = 0.4] (24){}
        (6.647, 3.386) node[fill = black, opacity = 0.4] (25){}
        (6.801, 2.593) node[fill = black, opacity = 0.4] (26){}
        (1.435, 1.512) node[fill = black, opacity = 0.4] (27){}
        (1.687, 6.188) node [fill = black](28){}
        (6.597, 5.672) node[fill = black, opacity = 0.4] (29){}
        (4.896, 1.102) node[fill = black, opacity = 0.4] (30){}
        (1.274, 1.987) node[fill = black, opacity = 0.4] (31){}
        (5.92, 5.031) node[fill = black] (32){}
        (1.909, 2.782) node [fill = black](33){}
        (4.266, 5.679) node [fill = black, opacity = 0.4](34){}
        (4.962, 5.572) node[fill = black, opacity = 0.4] (35){}
        (5.037, 7.254) node[fill = black, opacity = 0.4] (36){}
        (4.34, 4.261) node[fill = black, opacity = 0.4] (37){}
        (4.675, 4.815) node[fill = black] (38){}
        (6.02, 3.722) node [fill = black](39){}
        (0.271, -0.26) node[fill = black, opacity = 0.4] (40){}
        (5.203, 5.266) node[fill = black, opacity = 0.4] (41){}
        (5.813, 3.241) node[fill = black] (42){}
        (0.577, 2.03) node[fill = black, opacity = 0.4] (43){}
        (6.867, 4.454) node[fill = black, opacity = 0.4] (44){}
        (6.734, 1.979) node [fill = black, opacity = 0.4](45){}
        (4.192, 1.494) node[fill = black, opacity = 0.4] (46){}
        (7.537, 2.158) node [fill = black, opacity = 0.4](47){}
        (3.476, 3.797) node [fill = black](48){}
        (2.764, 3.977) node [fill = black](49){};
        \end{scope}
      \begin{scope}[thick]
        \draw (9) to (17);
        \draw (9) to (42);
        \draw (12) to (13);
        \draw (12) to (28);
        \draw (13) to (49);
        \draw (17) to (33);
        
        \draw (20) to (28);
        \draw (20) to (38);
        \draw (32) to (38);
        \draw (32) to (39);
        \draw (33) to (48);

        \draw (39) to (42);
        \draw (48) to (49);
      \end{scope}
        
      \begin{scope}[dotted]
        \draw (1) to (29);
        \draw (1) to (32);
        \draw (1) to (35);
        \draw (1) to (36);
        \draw (1) to (41);
        \draw (2) to (10);
        \draw (2) to (29);
        \draw (2) to (44);
        \draw (3) to (20);
        \draw (3) to (22);
        \draw (3) to (28);
        \draw (3) to (37);
        \draw (3) to (48);
        \draw (3) to (49);
        \draw (4) to (15);
        \draw (4) to (19);
        \draw (4) to (30);
        \draw (4) to (45);
        \draw (5) to (11);
        \draw (5) to (17);
        \draw (5) to (27);
        \draw (5) to (30);
        \draw (5) to (33);
        \draw (5) to (46);
        \draw (6) to (12);
        \draw (6) to (13);
        \draw (6) to (43);
        \draw (7) to (20);
        \draw (7) to (24);
        \draw (7) to (34);
        \draw (7) to (35);
        \draw (7) to (36);
        \draw (9) to (19);
        \draw (9) to (30);
        \draw (9) to (37);
        \draw (9) to (46);
        \draw (10) to (44);
        \draw (10) to (47);
        \draw (11) to (27);
        \draw (11) to (40);
        \draw (12) to (16);
        \draw (12) to (22);
        \draw (13) to (21);
        \draw (13) to (22);
        \draw (13) to (43);
        \draw (15) to (45);
        \draw (16) to (28);
        
        \draw (17) to (37);
        \draw (17) to (46);
         \draw (17) to (48);
        \draw (18) to (27);
        \draw (18) to (40);
        \draw (18) to (43);
        \draw (19) to (26);
        \draw (19) to (30);
        \draw (19) to (42);
        \draw (19) to (45);
        \draw (20) to (24);
        \draw (20) to (34);
         \draw (20) to (37);
        
        \draw (21) to (31);
        \draw (21) to (33);
        \draw (21) to (43);
        \draw (21) to (49);
        \draw (22) to (28);
        \draw (22) to (49);
        \draw (24) to (28);
        \draw (24) to (36);
        \draw (25) to (26);
        \draw (25) to (39);
        \draw (25) to (42);
        \draw (25) to (44);
        \draw (26) to (42);
        \draw (26) to (44);
        \draw (26) to (45);
        \draw (26) to (47);
        \draw (27) to (31);
        \draw (27) to (33);
        \draw (27) to (40);
        \draw (27) to (43);
        \draw (29) to (32);
        \draw (29) to (44);
        \draw (30) to (46);
        \draw (31) to (33);
        \draw (31) to (43);
        
        \draw (32) to (41);
         \draw (32) to (44);
         \draw (33) to (49);
        \draw (34) to (35);
        \draw (34) to (38);
        \draw (35) to (36);
        \draw (35) to (38);
        \draw (35) to (41);
         \draw (37) to (38);
        \draw (37) to (39);
        \draw (37) to (42);
        \draw (37) to (48);
        \draw (38) to (39);
        \draw (38) to (41);
         \draw (39) to (44);
        \draw (44) to (47);
        \draw (45) to (47);
        
      \end{scope}
    \end{tikzpicture}

    \caption[Cycle shortening flow frame 3]{After the second  step of the cycle shortening flow, for every face $f$ of $G$ with more edges than non-edges in the subgraph, the cycle $C_f$ is added to the subgraph. }
    \label{fig:CSF3}
\end{figure}

\begin{figure}
    \centering
  \begin{tikzpicture}[scale =1.5]
         \begin{scope}[circle, fill=black, draw=black, inner sep = 1.5pt]
      \draw[black]
        (6.036, 5.662) node[fill = black, opacity = 0.4](1){}
        (7.583, 5.73) node[fill = black, opacity = 0.4] (2){}
        (2.931, 5.148) node [fill = black, opacity = 0.4](3){}
        (6.079, 1.579) node[fill = black, opacity = 0.4] (4){}
        (3.997, 1.274) node [fill = black, opacity = 0.4](5){}
        (-0.012, 4.278) node[fill = black, opacity = 0.4] (6){}
        (4.406, 6.708) node[fill = black, opacity = 0.4] (7){}
        (4.398, 2.542) node[fill = black] (9){}
        (7.655, 3.346) node[fill = black, opacity = 0.4] (10){}
        (2.52, -0.353) node [fill = black, opacity = 0.4](11){}
        (0.791, 5.56) node[fill = black] (12){}
        (0.97, 3.766) node[fill = black] (13){}
        (7.826, 0.406) node [fill = black, opacity = 0.4](15){}
        (0.866, 6.721) node[fill = black, opacity = 0.4] (16){}
        (3.493, 2.663) node [fill = black](17){}
        (-0.11, 0.959) node[fill = black, opacity = 0.4] (18){}
        (6.016, 2.098) node[fill = black, opacity = 0.4] (19){}
        (3.528, 5.576) node[fill = black] (20){}
        (1.611, 2.899) node[fill = black, opacity = 0.4] (21){}
        (2.14, 5.139) node [fill = black, opacity = 0.4](22){}
        (3.159, 7.648) node[fill = black, opacity = 0.4] (24){}
        (6.647, 3.386) node[fill = black, opacity = 0.4] (25){}
        (6.801, 2.593) node[fill = black, opacity = 0.4] (26){}
        (1.435, 1.512) node[fill = black, opacity = 0.4] (27){}
        (1.687, 6.188) node [fill = black](28){}
        (6.597, 5.672) node[fill = black, opacity = 0.4] (29){}
        (4.896, 1.102) node[fill = black, opacity = 0.4] (30){}
        (1.274, 1.987) node[fill = black, opacity = 0.4] (31){}
        (5.92, 5.031) node[fill = black, opacity = 0.4] (32){}
        (1.909, 2.782) node [fill = black](33){}
        (4.266, 5.679) node [fill = black, opacity = 0.4](34){}
        (4.962, 5.572) node[fill = black, opacity = 0.4] (35){}
        (5.037, 7.254) node[fill = black, opacity = 0.4] (36){}
        (4.34, 4.261) node[fill = black, opacity = 0.4] (37){}
        (4.675, 4.815) node[fill = black] (38){}
        (6.02, 3.722) node [fill = black](39){}
        (0.271, -0.26) node[fill = black, opacity = 0.4] (40){}
        (5.203, 5.266) node[fill = black, opacity = 0.4] (41){}
        (5.813, 3.241) node[fill = black] (42){}
        (0.577, 2.03) node[fill = black, opacity = 0.4] (43){}
        (6.867, 4.454) node[fill = black, opacity = 0.4] (44){}
        (6.734, 1.979) node [fill = black, opacity = 0.4](45){}
        (4.192, 1.494) node[fill = black, opacity = 0.4] (46){}
        (7.537, 2.158) node [fill = black, opacity = 0.4](47){}
        (3.476, 3.797) node [fill = black](48){}
        (2.764, 3.977) node [fill = black](49){};
        \end{scope}
            \begin{scope}[thick]
        \draw (9) to (17);
        \draw (9) to (42);
        \draw (12) to (13);
        \draw (12) to (28);
        \draw (13) to (49);
        \draw (17) to (48);
        
        \draw (20) to (28);
        \draw (20) to (38);
        
        \draw (33) to (48);
        \draw (33) to (49);
        
        \draw (38) to (39);
        \draw (39) to (42);
        
      \end{scope}
        
      \begin{scope}[dotted]
        \draw (1) to (29);
        \draw (1) to (32);
        \draw (1) to (35);
        \draw (1) to (36);
        \draw (1) to (41);
        \draw (2) to (10);
        \draw (2) to (29);
        \draw (2) to (44);
        \draw (3) to (20);
        \draw (3) to (22);
        \draw (3) to (28);
        \draw (3) to (37);
        \draw (3) to (48);
        \draw (3) to (49);
        \draw (4) to (15);
        \draw (4) to (19);
        \draw (4) to (30);
        \draw (4) to (45);
        \draw (5) to (11);
        \draw (5) to (17);
        \draw (5) to (27);
        \draw (5) to (30);
        \draw (5) to (33);
        \draw (5) to (46);
        \draw (6) to (12);
        \draw (6) to (13);
        \draw (6) to (43);
        \draw (7) to (20);
        \draw (7) to (24);
        \draw (7) to (34);
        \draw (7) to (35);
        \draw (7) to (36);
        \draw (9) to (19);
        \draw (9) to (30);
        \draw (9) to (37);
        \draw (9) to (46);
        \draw (10) to (44);
        \draw (10) to (47);
        \draw (11) to (27);
        \draw (11) to (40);
        \draw (12) to (16);
        \draw (12) to (22);
        \draw (13) to (21);
        \draw (13) to (22);
        \draw (13) to (43);
        \draw (15) to (45);
        \draw (16) to (28);
         \draw (17) to (33);
        \draw (17) to (37);
        \draw (17) to (46);
         
        \draw (18) to (27);
        \draw (18) to (40);
        \draw (18) to (43);
        \draw (19) to (26);
        \draw (19) to (30);
        \draw (19) to (42);
        \draw (19) to (45);
        \draw (20) to (24);
        \draw (20) to (34);
         \draw (20) to (37);
        
        \draw (21) to (31);
        \draw (21) to (33);
        \draw (21) to (43);
        \draw (21) to (49);
        \draw (22) to (28);
        \draw (22) to (49);
        \draw (24) to (28);
        \draw (24) to (36);
        \draw (25) to (26);
        \draw (25) to (39);
        \draw (25) to (42);
        \draw (25) to (44);
        \draw (26) to (42);
        \draw (26) to (44);
        \draw (26) to (45);
        \draw (26) to (47);
        \draw (27) to (31);
        \draw (27) to (33);
        \draw (27) to (40);
        \draw (27) to (43);
        \draw (29) to (32);
        \draw (29) to (44);
        \draw (30) to (46);
        \draw (31) to (33);
        \draw (31) to (43);
        
        \draw (32) to (41);
         \draw (32) to (44);
         \draw (32) to (38);
         \draw (32) to (39);
         
        \draw (34) to (35);
        \draw (34) to (38);
        \draw (35) to (36);
        \draw (35) to (38);
        \draw (35) to (41);
         \draw (37) to (38);
        \draw (37) to (39);
        \draw (37) to (42);
        \draw (37) to (48);
        
        \draw (38) to (41);
         \draw (39) to (44);
        \draw (44) to (47);
        \draw (45) to (47);
        \draw (48) to (49);
      \end{scope}
    \end{tikzpicture}

    \caption[Cycle shortening flow frame 4]{After the third step of the cycle shortening flow, for every face $f$ of $G$ with more edges than non-edges in the subgraph, the cycle $C_f$ is added to the subgraph. }
    \label{fig:CSF4}
\end{figure}
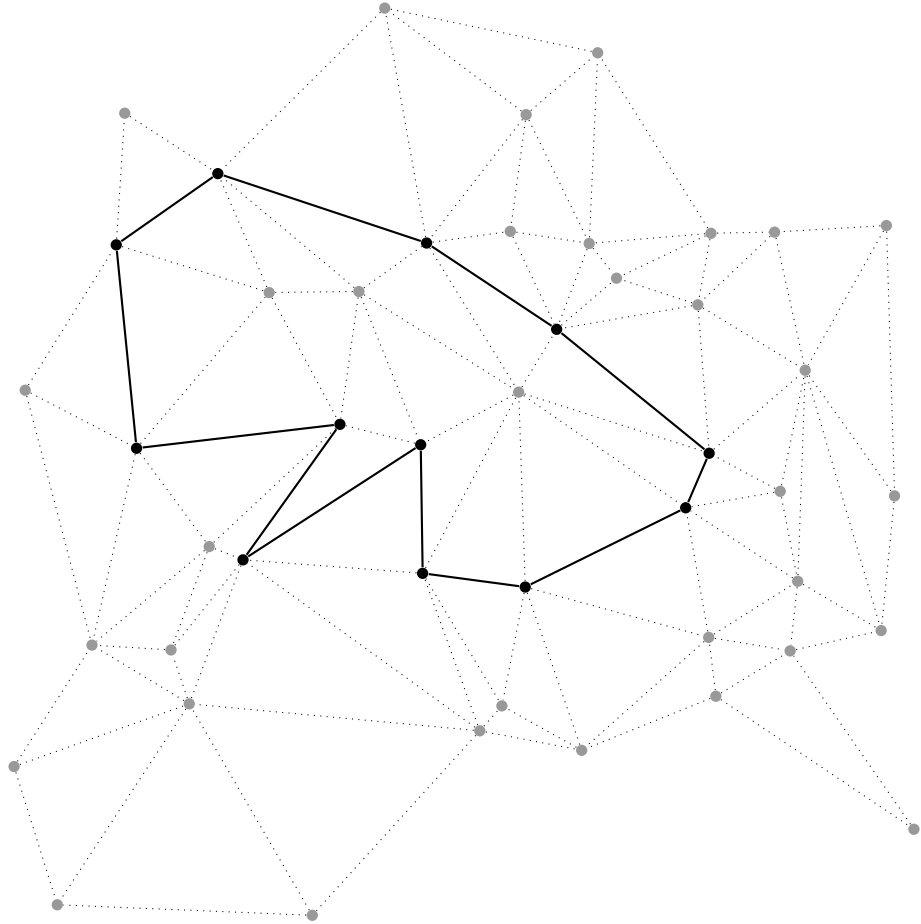

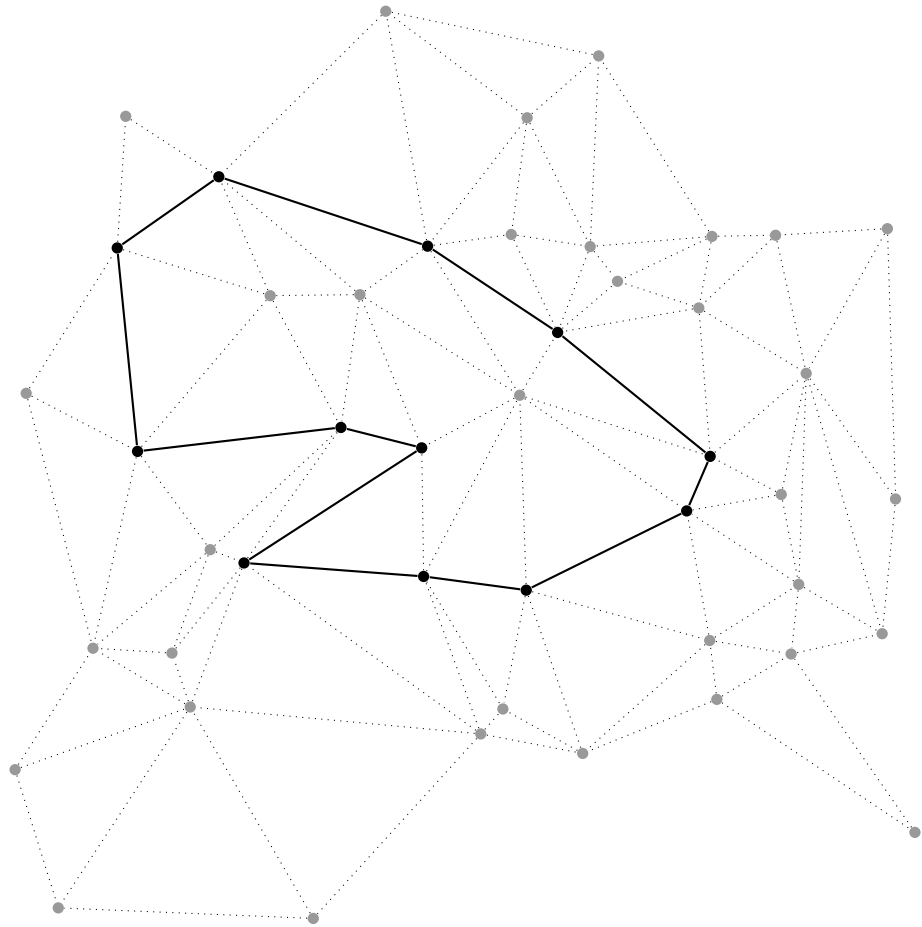
\begin{figure}
    \centering
  \begin{tikzpicture}[scale =1.5]
         \begin{scope}[circle, fill=black, draw=black, inner sep = 1.5pt]
      \draw[black]
        (6.036, 5.662) node[fill = black, opacity = 0.4](1){}
        (7.583, 5.73) node[fill = black, opacity = 0.4] (2){}
        (2.931, 5.148) node [fill = black, opacity = 0.4](3){}
        (6.079, 1.579) node[fill = black, opacity = 0.4] (4){}
        (3.997, 1.274) node [fill = black, opacity = 0.4](5){}
        (-0.012, 4.278) node[fill = black, opacity = 0.4] (6){}
        (4.406, 6.708) node[fill = black, opacity = 0.4] (7){}
        (4.398, 2.542) node[fill = black] (9){}
        (7.655, 3.346) node[fill = black, opacity = 0.4] (10){}
        (2.52, -0.353) node [fill = black, opacity = 0.4](11){}
        (0.791, 5.56) node[fill = black] (12){}
        (0.97, 3.766) node[fill = black] (13){}
        (7.826, 0.406) node [fill = black, opacity = 0.4](15){}
        (0.866, 6.721) node[fill = black, opacity = 0.4] (16){}
        (3.493, 2.663) node [fill = black](17){}
        (-0.11, 0.959) node[fill = black, opacity = 0.4] (18){}
        (6.016, 2.098) node[fill = black, opacity = 0.4] (19){}
        (3.528, 5.576) node[fill = black] (20){}
        (1.611, 2.899) node[fill = black, opacity = 0.4] (21){}
        (2.14, 5.139) node [fill = black, opacity = 0.4](22){}
        (3.159, 7.648) node[fill = black, opacity = 0.4] (24){}
        (6.647, 3.386) node[fill = black, opacity = 0.4] (25){}
        (6.801, 2.593) node[fill = black, opacity = 0.4] (26){}
        (1.435, 1.512) node[fill = black, opacity = 0.4] (27){}
        (1.687, 6.188) node [fill = black](28){}
        (6.597, 5.672) node[fill = black, opacity = 0.4] (29){}
        (4.896, 1.102) node[fill = black, opacity = 0.4] (30){}
        (1.274, 1.987) node[fill = black, opacity = 0.4] (31){}
        (5.92, 5.031) node[fill = black, opacity = 0.4] (32){}
        (1.909, 2.782) node [fill = black](33){}
        (4.266, 5.679) node [fill = black, opacity = 0.4](34){}
        (4.962, 5.572) node[fill = black, opacity = 0.4] (35){}
        (5.037, 7.254) node[fill = black, opacity = 0.4] (36){}
        (4.34, 4.261) node[fill = black, opacity = 0.4] (37){}
        (4.675, 4.815) node[fill = black] (38){}
        (6.02, 3.722) node [fill = black](39){}
        (0.271, -0.26) node[fill = black, opacity = 0.4] (40){}
        (5.203, 5.266) node[fill = black, opacity = 0.4] (41){}
        (5.813, 3.241) node[fill = black] (42){}
        (0.577, 2.03) node[fill = black, opacity = 0.4] (43){}
        (6.867, 4.454) node[fill = black, opacity = 0.4] (44){}
        (6.734, 1.979) node [fill = black, opacity = 0.4](45){}
        (4.192, 1.494) node[fill = black, opacity = 0.4] (46){}
        (7.537, 2.158) node [fill = black, opacity = 0.4](47){}
        (3.476, 3.797) node [fill = black](48){}
        (2.764, 3.977) node [fill = black](49){};
        \end{scope}
            \begin{scope}[thick]
        \draw (9) to (17);
        \draw (9) to (42);
        \draw (12) to (13);
        \draw (12) to (28);
        \draw (13) to (49);

        \draw (20) to (28);
        \draw (20) to (38);
        
        \draw (33) to (48);

        \draw (38) to (39);
        \draw (39) to (42);

        \draw (17) to (33);
        \draw (48) to (49);
      \end{scope}
        
      \begin{scope}[dotted]
        \draw (17) to (48);
        \draw (33) to (49);
      
        \draw (1) to (29);
        \draw (1) to (32);
        \draw (1) to (35);
        \draw (1) to (36);
        \draw (1) to (41);
        \draw (2) to (10);
        \draw (2) to (29);
        \draw (2) to (44);
        \draw (3) to (20);
        \draw (3) to (22);
        \draw (3) to (28);
        \draw (3) to (37);
        \draw (3) to (48);
        \draw (3) to (49);
        \draw (4) to (15);
        \draw (4) to (19);
        \draw (4) to (30);
        \draw (4) to (45);
        \draw (5) to (11);
        \draw (5) to (17);
        \draw (5) to (27);
        \draw (5) to (30);
        \draw (5) to (33);
        \draw (5) to (46);
        \draw (6) to (12);
        \draw (6) to (13);
        \draw (6) to (43);
        \draw (7) to (20);
        \draw (7) to (24);
        \draw (7) to (34);
        \draw (7) to (35);
        \draw (7) to (36);
        \draw (9) to (19);
        \draw (9) to (30);
        \draw (9) to (37);
        \draw (9) to (46);
        \draw (10) to (44);
        \draw (10) to (47);
        \draw (11) to (27);
        \draw (11) to (40);
        \draw (12) to (16);
        \draw (12) to (22);
        \draw (13) to (21);
        \draw (13) to (22);
        \draw (13) to (43);
        \draw (15) to (45);
        \draw (16) to (28);
         
        \draw (17) to (37);
        \draw (17) to (46);
         
        \draw (18) to (27);
        \draw (18) to (40);
        \draw (18) to (43);
        \draw (19) to (26);
        \draw (19) to (30);
        \draw (19) to (42);
        \draw (19) to (45);
        \draw (20) to (24);
        \draw (20) to (34);
         \draw (20) to (37);
        
        \draw (21) to (31);
        \draw (21) to (33);
        \draw (21) to (43);
        \draw (21) to (49);
        \draw (22) to (28);
        \draw (22) to (49);
        \draw (24) to (28);
        \draw (24) to (36);
        \draw (25) to (26);
        \draw (25) to (39);
        \draw (25) to (42);
        \draw (25) to (44);
        \draw (26) to (42);
        \draw (26) to (44);
        \draw (26) to (45);
        \draw (26) to (47);
        \draw (27) to (31);
        \draw (27) to (33);
        \draw (27) to (40);
        \draw (27) to (43);
        \draw (29) to (32);
        \draw (29) to (44);
        \draw (30) to (46);
        \draw (31) to (33);
        \draw (31) to (43);
        
        \draw (32) to (41);
         \draw (32) to (44);
         \draw (32) to (38);
         \draw (32) to (39);
         
        \draw (34) to (35);
        \draw (34) to (38);
        \draw (35) to (36);
        \draw (35) to (38);
        \draw (35) to (41);
         \draw (37) to (38);
        \draw (37) to (39);
        \draw (37) to (42);
        \draw (37) to (48);
        
        \draw (38) to (41);
         \draw (39) to (44);
        \draw (44) to (47);
        \draw (45) to (47);
        
      \end{scope}
    \end{tikzpicture}

    \caption[Cycle shortening flow frame 5]{After the fourth step of the cycle shortening flow, the subgraph is close to a local minimum. This is not a proper local minimum but the only shortening step is back to frame \ref{fig:CSF4}. This means that the flow ends in a 2-cycle and neither strategy profile is a Nash equilibrium. In the mixed strategy case, a Nash equilibrium may be found between these two strategy profiles.  }
    \label{fig:CSF5}
\end{figure}
\end{document}